%% file: main.tex
\documentclass[a4paper,11pt]{article}
\usepackage{jheppub} 
\usepackage{lineno}
\usepackage{xcolor}
\usepackage{xfrac}
\usepackage{amsmath}
\usepackage{amsfonts}
\usepackage{amsthm}
\usepackage{amssymb}
\usepackage{accents}
\usepackage{proof-at-the-end}
\usepackage{feynmf}
\usepackage{tikz}
\usetikzlibrary{shapes.geometric, arrows, patterns}
\usetikzlibrary{calc}
\tikzstyle{top} = [rectangle, rounded corners, minimum width=3cm, minimum height=1cm, text centered, text width=3cm, draw=black, fill=red!30]
\allowdisplaybreaks

\def\bnon{\begin{equation*}}
\def\enon{\end{equation*}}
\def\bnu{\begin{equation}}
\def\enu{\end{equation}}
\def\bqa{\begin{eqnarray*}}
\def\eqa{\end{eqnarray*}}
\def\blsd{\begin{enumerate}}
\def\elsd{\end{enumerate}}
\def\blst{\begin{itemize}}
\def\elst{\end{itemize}}

\def\TT{\widetilde{\mathcal{ T}}}
\def\T{\mathcal{T}}
\def\P{\mathcal{P}}

\def\rh{\mathrm{H}}

\def\gai{\gamma_{\mathrm{AI}}}

\def\ga{\Gamma}

\def\hopf{\mathcal{H}_{Q}}

\def\eps{\varepsilon}
\def\ba{\bar{a}}

\def\ide{\mathbb{I}}

\newtheorem{theo}{Theorem}
\newtheorem{lem}[theo]{Lemma}
\newtheorem{pro}{Proposition}

\newtheorem{cor}[theo]{Corollary}

\def\A{\mathcal{A}}
\def\I{\mathcal{I}}

\input{pics}

\preprint{DESY-24-127}

\title{\boldmath The asymptotic Hopf Algebra of Feynman Integrals}

\author[a]{Mrigankamauli Chakraborty}
\emailAdd{mrigankamauli.chakraborty@desy.de}
\affiliation[a]{II. Institute for Theoretical Physics, Hamburg University\\
Luruper Chaussee 149, D-22761 Hamburg, Germany}

\author[b]{and Franz Herzog}
\emailAdd{fherzog@ed.ac.uk}
\affiliation[b]{Higgs Centre for Theoretical Physics, School of Physics and Astronomy\\
The University of Edinburgh, Edinburgh EH9 3FD, Scotland, UK}

\begin{document}

\abstract{
The method of regions is an approach for developing asymptotic expansions of Feynman Integrals. 
We focus on expansions in Euclidean signature, where the method of regions can also be formulated as an expansion by subgraph. We show that for such expansions valid around small/large masses and momenta the graph combinatorial operations can be formulated in terms of what we call the asymptotic Hopf algebra. This Hopf algebra is closely related to the motic Hopf algebra underlying the $R^*$ operation, an extension of Bogoliubov's $R$ operation, to subtract both IR and UV divergences of Feynman integrals in the Euclidean. We focus mostly on the leading power, for which the Hopf algebra formulation is simpler. We uncover a close connection between Bogoliubov's $R$ operation in the Connes-Kreimer formulation and the remainder $\mathcal{R}$ of the series expansion, whose Hopf algebraic structure is identically formalised in the corresponding group of characters. While in the Connes-Kreimer formulation the UV counterterm is formalised in terms of a twisted antipode, we show that in the expansion by subgraph a similar role is played by the integrand Taylor operator. To discuss the structure of higher power expansions we introduce a novel Hopf monoid formulation. }

\maketitle
\flushbottom

\section{Introduction}

Feynman integrals are the building blocks of S-matrix elements and Green's functions in Quantum Field Theory. As such they are important for theoretical predictions in collider experiments. However, their mathematical structure is also interesting in their own right and has inspired developments in mathematics. One such development is the Connes-Kreimer Hopf algebra of Feynman graphs \cite{Kreimer:1997dp, Connes:1999zw, Connes:1999yr} which governs the graph combinatorics of Bogoliubov's $R$ operation \cite{Bogoliubov:1957gp, Hepp:1966eg, Zimmermann:1969jj}. More recently this construction has also been extended to the class of Euclidean IR divergences, via the Hopf algebra of motic graphs \cite{Brown:2015fyf,Beekveldt:2020kzk}, and the corresponding $R^*$-operation \cite{Chetyrkin:1982nn, Chetyrkin:1983wh, Chetyrkin:1984xa, Larin:2002sc, Chetyrkin:2017ppe, Herzog:2017bjx}. 

Hopf algebraic structures have also appeared in other contexts within perturbative Quantum Field Theory. Well known examples are the Hopf algebra structures on motivic Zeta values \cite{Brown:2011ik} and generalised polylogarithms \cite{Goncharov:2001iea} in which a large class of Feynman integrals can be expressed and which form the basis of the symbol-calculus \cite{Goncharov:2010jf, Duhr:2012fh}. This Hopf algebra has also been formalised in terms of a diagrammatic coaction at the level of Feynman diagrams \cite{Abreu:2021vhb,Abreu:2017mtm}, and has given rise to a new mysterious antipode relation among amplitudes \cite{Dixon:2022xqh}. A new Hopf algebra has recently been uncovered in connection to the kinematic algebra appearing in the double-copy gauge-gravity correspondence \cite{Brandhuber:2022enp, Chen:2022nei}. Hopf algebras are therefore playing an increasingly important role in our mathematical understanding of perturbative Quantum Field Theory, and seem to provide the key to handle and manage the complicated structures appearing in scattering amplitudes.

In this paper we will focus on series or asymptotic expansions of Feynman integrals around potentially non-regular kinematic points. There exists a general approach for developing such series known as the \emph{method of regions} \cite{Smirnov:1990rz, Smirnov:1994tg, Beneke:1997zp, Tkachov:1997ap, Smirnov:1998vk, Smirnov:1999bza, Smirnov:2002pj, Jantzen:2011nz, Semenova:2018cwy}. The problem is that the naive Taylor expansion at the level of the integrand of a Feynman integral does not generally yield a correct answer. The essence of the method of regions is that the full answer can be recovered by adding to this naive expansion also expansions around several so-called \emph{regions}. Each region can be identified with a certain scaling of the integration variables with an expansion parameter, $\lambda$, to some power. The difficult part in general is to identify the correct set of regions required for the particular expansion. In general this is non-trivial, and there exist no general proof for this procedure to work. 

The situation is somewhat improved for certain kinematic expansions. One of these includes expansions of completely offshell Feynman diagrams, which can be analytically continued into the Euclidean region. More specifically these include the expansions around small and/or large external momenta or internal masses. This set of region expansions is also special in that it can be formulated in a graph-theoretic language, in the sense that the regions are identified with certain sets of subgraphs \cite{ Chetyrkin:1982zq, Chetyrkin:1983qlc, Gorishnii:1983su, Gorishnii:1986mv, LlewellynSmith:1987jx, Chetyrkin:1988zzzzz, Chetyrkin:1988zz, Chetyrkin:1988cu, Gorishnii:1989dd}. Hence also the name \emph{expansion by subgraph}. A proof for the validity in this case was provided by Smirnov \cite{Smirnov:1990rz, Smirnov:1994tg}.

Going beyond the Euclidean case there exists a geometric method, applicable in the parametric representation, whenever the $\mathcal{F}$ or 2nd Symanzik polynomial is of definite sign. This method makes use of a tropical geometry associated to the Feynman polytope, first discovered in the context of Sector decomposition \cite{Kaneko:2009qx}, before it was extended and implemented as a way to unveil regions \cite{ Pak:2010pt, Jantzen:2012mw, Ananthanarayan:2018tog, Heinrich:2021dbf}. An expansion-by-subgraph approach valid for Minkowskian expansions has also recently been developed for the onshell-expansion \cite{Gardi:2022khw}, the soft expansion \cite{Herzog:2023sgb} and a wider class of related expansions \cite{Ma:2023hrt}.

The main focus of this paper is to provide a new Hopf algebraic  formulation of the expansion by region, valid for the Euclidean Case. We will thus connect two areas: renormalization Hopf algebras and the expansion by subgraph framework. The crucial feature, allowing us to make this connection, is that the subgraphs, appearing in the above mentioned $R^*$ operation, are closely  related to the \emph{asymptotic subgraphs} appearing in the expansion by subgraph. One of the main results of this paper is to show that these asymptotic subgraphs indeed form a Hopf algebra, and that the expansion by subgraph can be formulated naturally in this framework. A special role will be played by the remainder of the expansion which we present in a manner closely related to Bogoliubov's $R$-operation in the Connes-Kreimer formulation, i.e. as a convolution product in the corresponding group of characters where the counterterm operation, or twisted antipode, is now identified with the Taylor operator. We also motivate this construction by making use of a Birkhoff decomposition. While our discussion is mostly centered on the leading logarithmic expansion case, we do also introduce a Hopf monoid \cite{aguiar2017hopfmonoidsgeneralizedpermutahedra,monoidbook} formulation which allows to keep track of the labels of edges in the tensor product. This provides an alternative route to discuss the Hopf structure at the level of the integrand, and at higher powers in the expansion.

The paper is structured as follows. We review the necessary background in section \ref{bg}. The asymptotic Hopf algebra is constructed in section \ref{hopfdef} purely at the level of graphs. The Hopf algebraic formulation of the expansion by subgraph is provided in section \ref{hopfformulation}. We conclude in section \ref{conclusion}.



\section{Background}\label{bg}



In the following we introduce notation and review the main concepts of the expansion by subgraph as well as the basics of the diagrammatic Hopf algebra of renormalisation.   

\subsection{Notation}
We work exclusively with Feynman integrals, with off-shell  non-exceptional external momenta, which can be analytically continued into the Euclidean region. The kinematic data of the graph $\ga$ will be specified as follows:
\begin{itemize}
    \item $\{p_i\}$ is the set of soft momenta,
    \item $\{q_i\}$ is the set of hard momenta,
    \item $\{l_i\}$ is the set of loop momenta,
    \item $\{m_i\}$ is the set of soft internal masses,
    \item and $\{M_i\}$ is the set of hard internal masses.
\end{itemize}

It will be convenient to let $A$ be the algebra\footnote{Note that the set of functions are naturally closed under addition and multiplication, and thus form an algebra.} of analytic functions defined on the momenta and the dimensional regularisation parameter $\epsilon=\frac{4-D}{2}$, with $D$ the space-time dimension. With $\mathcal{H}$ as the set of Feynman graphs in a theory, the map $\phi:\mathcal{H}\to A$ then provides the analytic expression for the Feynman integral corresponding to each graph in $\mathcal{H}$. We further decompose $\phi$ as a composition of two operations:
\bnon
\phi = \mathcal{I}\circ F\,,
\enon 
where $F:\mathcal{H}\to \mathcal{A}$ applies the Feynman rules to the graph to make the integrand. $\mathcal{A}$ is then the set of integrands corresponding to the graphs in $\mathcal{H}$. The integration operation $\mathcal{I}:\mathcal{A}\to A$ integrates the integrand w.r.t. the loop momenta.\\

\subsection{Asymptotic expansion and regular Taylor operator}\label{lambda}

Consider now the expansion of a Feynman integral $\phi(\Gamma)$ around a set of \emph{soft} parameters, i.e. small masses and momenta $\{m_1,\dots,p_1,\dots\}$. It is convenient to rescale the soft parameters with a book-keeping parameter $\lambda$, i.e. $\{m_1,\dots,p_1,\dots\}\to \{\lambda m_1,\dots,\lambda p_1,\dots\}$, such that 
$$
\phi(\Gamma)\to \phi(\Gamma;\lambda)\,,\qquad  \phi(\Gamma)=\phi(\Gamma;\lambda)\Big|_{\lambda=1}\,.
$$ 
Effectively, we can view the multi-variable expansion as a series expansion in the single variable $\lambda$, which can be set to $1$ at the end of a calculation. In general, due to the appearance of singularities and/or discontinuities, in the limit $\lambda\to 0$ the series expansion of $\phi(\Gamma;\lambda)$ is not regular and takes the general form:
\begin{align*}
\phi(\Gamma;\lambda)&=\sum_{r} \lambda^{k_r-2r\eps} \phi_r(\Gamma;\lambda)\,,
\end{align*}
where, as it turns out, $r\in\{0,\dots,L\}$ with $L$ the number of loops of $\Gamma$. In contrast to $\phi(\Gamma;\lambda)$, the different $\phi_r(\Gamma;\lambda)$ are regular at $\lambda=0$. One may therefore define an $n$th order \emph{regular} Taylor expansion operator, namely $\TT^{n}$, of $\phi(\Gamma;\lambda)$ (which after expanding $\lambda^{-r\eps}=\exp(-r\eps\log\lambda)$ in $\eps$ leads to a log power series in $\lambda$) via the ordinary $n$th order Taylor expansion operator $\T^{n}$ of $\phi_r(\Gamma;\lambda)$, as follows,
\begin{align*}    \TT^{n}\phi(\Gamma;\lambda):=\sum_{r}\lambda^{k_r-2r\eps} \,
    \T^{n-k_r} \phi_r(\Gamma;\lambda)\,,
\end{align*}
where $k_r$ is a non-positive integer, and
\bnon
\T^{n} f(\lambda)=\sum_{k=0}^n \T^{(n)} f(\lambda)\,, \qquad 
\T^{(n)} f(\lambda)= \frac{\lambda^{n}}{n!}\frac{d^nf(\lambda)}{d\lambda^n}\Big|_{\lambda=0}\,.
\enon
Note that later on we will also make use of the notation $\TT^{n}=\sum_n\TT^{(n)}$.

It will be important in the following to distinguish clearly between 
acting the Taylor operator on the integrand or the integral. 
For this purpose we define a further Taylor operation, 
$$
T^{(n)}:\mathcal{H}\to A\,, \qquad T^{(n)} = \mathcal{I}\circ \T^{(n)}\circ F\,,
$$ 
which acts directly on the Feynman diagram, Taylor expands the integrand, and then integrates every term of the expansion w.r.t. the loop momenta. We furthermore define a second Taylor operator $\widetilde{T}$ which also acts on the graph but \emph{regular} Taylor expands at the level of the integral:
\bnon
\widetilde{T}^{(n)}:\mathcal{H}\to A\,,\qquad \widetilde{T}^{(n)} = \widetilde{\T}^{(n)}\circ\mathcal{I}\circ F\,.
\enon
This operator gives us the value of the integral at a certain order in the soft scales. Thus, it performs the expansion by regions in a particular asymptotic limit. 

\subsection{Expansion by subgraph}
\noindent We will now outline the procedure for expansion by subgraph as developed by Smirnov \cite{Smirnov:1990rz, Smirnov:1994tg, Smirnov:2002pj}. Given a graph $\ga$, a subgraph $\gamma$ of $\ga$ is called \textit{asymptotically irreducible} (AI) if
\blsd
\item $\gamma$ contains all the  hard legs ($q$) and heavy lines (lines with an $M$),
\item the graph obtained from $\gamma$, by contracting all heavy lines and identifying\footnote{In Graph theory to identify vertices $v_1,v_2,\dots$ is to replace them with a single vertex $v$, such that all lines which were incident to $v_1,v_2,\dots$ are incident to $v$.} all the external vertices associated with the hard legs in $\gamma$,  is componentwise 1PI. 
\elsd

We will denote AI subgraphs with $\gai$. The contracted graphs $\ga/\gai$ are obtained by contracting all the internal edges of $\gai$ in $\ga$. An example will be given shortly in section \ref{naive}. 

With these definitions we can now formulate the asymptotic expansion of the graph $\ga$ to order $a$:
\bnu \label{ebrnew}
\widetilde{T}^{a}(\ga) = \sum_{\gai\subseteq \ga} \mathcal{I}\big( T^{\ba+\omega}(\gai)\,F(\ga/\gai) \big)\,,
\enu
where
\begin{itemize}
\item $\bar{a} = a-\Omega$,
\item $\Omega$ is the superficial degree of divergence (SDD) of $\ga$,
\item $\omega$ is the SDD of $\gai$,
\item the operator $T^{\ba+\omega}$ expands around all masses of, and external momenta to, $\gai$ which are not hard scales of $\Gamma$.
\end{itemize}
Note also that soft external momenta of $\gai$ can indeed correspond to internal momenta of $\ga$. 

The convergence of the expansion is assured by Smirnov's asymptotic theorem \cite{Smirnov:1990rz, Smirnov:1994tg}, which states that the remainder given by 
\bnu\label{remo}\mathcal{R}^a(\ga) := \phi(\ga)-\widetilde{T}^{a}(\ga)\enu 
is $\mathcal{O}(\lambda^{a+1})$ modulo logarithms.

\subsection{One-loop triangle diagram}\label{naive}

\noindent For the sake of an example, consider the one-loop triangle:
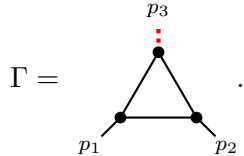
\begin{center}
$\Gamma = $
\begin{minipage}{2.2cm}
\begin{tikzpicture}
\draw [thick] (0,0.866) -- (-0.5,0);
\draw [ultra thick, dotted, red] (0,0.866) -- (0,1.216);
\draw [thick] (0.5,0) -- (0,0.866);
\draw [thick] (0.5,0) -- (0.75,-0.25);
\draw [thick] (-0.5,0) -- (-0.75,-0.25);
\draw [thick] (-0.5,0) -- (0.5,0);
\filldraw (-0.5,0) circle (2pt);
\filldraw (0.5,0) circle (2pt);
\filldraw (0, 0.866) circle (2pt);
\draw (0,1.416) node{${\scriptstyle {p_3}}$};
\draw (0.9,-0.4) node{${\scriptstyle {p_2}}$};
\draw (-0.9,-0.4) node{${\scriptstyle {p_1}}$};
\end{tikzpicture}
\end{minipage}.
\end{center}
The dotted red line is considered to carry a small external momentum and all internal lines are massless. The Feynman integral for this graph is
\bnon
\phi(\Gamma)=\int \frac{d^Dk}{i\pi^{D/2}}\;\frac{1}{k^2(k+p_3)^2(k-p_1)^2}\,,
\enon
and the integral is dimensionally regularised with $D=4-2\varepsilon$. In  \cite{Chavez:2012kn} an analytic expression of the integral in the Euclidean region at the leading order in $\varepsilon$ was given as:
\bnon
    (-p_1^2)^{-1}\;\frac{4i\, \mathcal{P}_2(z)}{z-\bar{z}}
\enon
where $\mathcal{P}_2(z)$ is Zagier's single-valued dilogarithm, defined as
\bnon
    \mathcal{P}_2(z) = \text{Im}\left\{\text{Li}_2(z) - \ln|z|\,\text{Li}_1(z)\right\}\,,
\enon
and
\bnon
    z = 1 + u - v + \sqrt{\lambda(1,u,v)}\,,
\qquad
    \bar{z} = 1 + u - v - \sqrt{\lambda(1,u,v)}\,,
\enon
\bnon
    u \equiv \frac{p_3^2}{p_1^2} \quad , \quad v \equiv \frac{p_2^2}{p_1^2}\,,
\enon
where
\bnon
\lambda(a,b,c) = a^2 + b^2 + c^2 -2ab -2bc -2ca\,.
\enon
Note that $\bar{z}$ corresponds to the complex conjugate of $z$ in the Euclidean regime where $\lambda(1,u,v)<0$, but this does not necessarily hold in other regimes. The expressions for $z$ and $\bar{z}$ for Euclidean $p_3^2$ and $p_1^2$ are 
\bnon
z =  -\frac{|p_3|}{|p_1|}e^{-i\theta},\; \bar{z} = -\frac{|p_3|}{|p_1|}e^{i\theta}\quad \text{where}\; \theta = \cos^{-1}\left(\frac{|p_3.p_1|}{|p_1||p_3|}\right)\,,
\enon
  and $|p_i| = (-p_i^2)^{1/2}$.
 $\bar{z}$ being the complex conjugate of $z$, crucially depends on the fact that the Cauchy-Schwarz inequality applies between Euclidean momenta. Thus, $\theta$ as defined above is real. 

As $|p_3|<|p_1|$, we have $|z|<1$, and thus Li$_2(z)$ and Li$_1(z)$ can be Taylor expanded around $z=0$. Thus we can obtain the integral as an infinite series in the small ratio of $\frac{|p_3|}{|p_1|}$.
The first 3 terms  in the expansion yield:
\begin{align}
\label{eq:trionelooexp}
 \phi(\Gamma;\lambda)\Big|_{\substack{p_1^2=1,\\\eps=0}} =
 &
 2-\log (z\bar z) + 
 \frac{1}{2}(z+\bar z)\big(1-\log (z\bar z)\big)\, \lambda \nonumber\\
&+\frac{1}{9} (z^2 + z \bar z + \bar z^2)\big(2-3\log (z\bar z)\big)\,\lambda^2
+\mathcal{O}(\lambda^3) 
\end{align}
where $\lambda\sim p_3^\mu$ is a book-keeping parameter for the expansion, introduced in section \ref{lambda}.
On the other hand, using expansion by subgraph gives us the following expression for the $n$th order term in the expansion:
\bnon
    \widetilde{T}^{(n)}_{\{p_3\}}(\ga) 
    = \; T_{\{p_3\}}^{(n)}\left( \begin{minipage}{2.5cm}
\begin{tikzpicture}
\draw [thick] (0,0.866) -- (-0.5,0);
\draw [ultra thick, dotted, red] (0,0.866) -- (0,1.216);
\draw [thick] (0.5,0) -- (0,0.866);
\draw [thick] (0.5,0) -- (0.75,-0.25);
\draw [thick] (-0.5,0) -- (-0.75,-0.25);
\draw [thick] (-0.5,0) -- (0.5,0);
\filldraw (-0.5,0) circle (2pt);
\filldraw (0.5,0) circle (2pt);
\filldraw (0, 0.866) circle (2pt);
\draw (0,1.416) node{${\scriptstyle {p_3}}$};
\draw (0.9,-0.4) node{${\scriptstyle {p_2}}$};
\draw (-0.9,-0.4) node{${\scriptstyle {p_1}}$};
\end{tikzpicture}
\end{minipage}\right)
+ T_{\{k,p_3\}}^{(n)}\left(\begin{minipage}{2.8cm}
\begin{tikzpicture}
\draw (0,0) -- (-1,0);
\draw (0,0) -- (0.35,-0.35);
\draw (-1,0) -- (-1.35,-0.35);
\draw [ultra thick, dotted, red](-1,0) -- (-1.35,0.35);
\draw [ultra thick, dotted, red](0,0) -- (0.35,0.35);
\filldraw (-1,0) circle (2pt);
\filldraw (0,0) circle (2pt);
\draw (0.55,-0.55) node{${\scriptstyle {p_2}}$};
\draw (-1.55,-0.55) node{${\scriptstyle {p_1}}$};
\draw (0.55,0.55) node{${\scriptstyle {k+p_3}}$};
\draw (-1.55,0.55) node{${\scriptstyle {k}}$};
\end{tikzpicture}
\end{minipage}\right)\ast
\phi \left(\begin{minipage}{2.5cm}
\begin{tikzpicture}
\draw[ thick, -] (0,0.5) arc (90:450:0.5);
\draw [ultra thick, dotted, red] (-1,0) -- (-0.5,0);
\draw (0.5,0) -- (0.85,0.35);
\draw (0.5,0) -- (0.85,-0.35);
\filldraw (-0.5,0) circle (2pt);
\filldraw [color=black, fill=white](0.5,0) circle (3pt);
\draw (-1.2,0.216) node{${\scriptstyle {p_3}}$};
\draw (1.05,0.55) node{${\scriptstyle {p_2}}$};
\draw (1.05,-0.55) node{${\scriptstyle {p_1}}$};
\end{tikzpicture}
\end{minipage}\quad\right)
\enon 
The hollow vertex in the contracted graph shows where the monomial from the Taylor expansion is inserted. For the zeroth order expansion there is no monomial to insert, but there is at higher orders. Evaluating the first three orders one obtains
\bqa
 \widetilde{T}^{(0)}_{\{p_3\}}(\ga)\Big|_{\substack{p_1^2=1,\\\eps=0}}     & = & \left(\frac{1}{\varepsilon}+2  - \log(z \bar z)\right) + \left(-\frac{1}{\varepsilon} \right)  =   2-\log(z \bar z)  \\
\widetilde{T}^{(1)}_{\{p_3\}}(\ga)\Big|_{\substack{p_1^2=1,\\\eps=0}}  & = &  \left(\frac{1}{2}(z+\bar z)\big(
\frac{1}{\eps}+2-\log (z\bar z)\big) \right) +\left(\frac{1}{2}(z+\bar z)\big(-1- \frac{1}{\eps}\big)\right) \\
 \widetilde{T}^{(2)}_{\{p_3\}}(\ga)\Big|_{\substack{p_1^2=1,\\\eps=0}}  & = &  
 \left(
 \frac{1}{3}(z^2+\bar z^2+z \bar z)\big( \frac{1}{\eps}-\log (z\bar z) \big) +\frac{13}{18}(z^2+\bar z^2)+\frac{5}{9}z\bar z
 \right) \\
&&
 +\left(  -\frac{1}{3\eps}(z^2+\bar z^2+z \bar z)  -\frac{1}{2}(z^2+\bar z^2)-\frac{1}{3}z\bar z   \right) 
\eqa
which agrees with eq.~(\ref{eq:trionelooexp}). While the total result is finite the two region contributions contain poles in $\epsilon$ which cancel in the sum.

\subsection{Renormalisation and related diagrammatic Hopf algebras}
\label{sec:conneskreimer}
In this section we will review diagrammatic Hopf algebras that have been used in the context of renormalisation and IR subtraction. In refs. \cite{Kreimer:1997dp,Connes:1999zw,Connes:1999yr,Connes:2000fe} Connes and Kreimer founded the Hopf algebraic structure of renormalisation, where the elements are \emph{bridgeless subgraphs}; these are possibly disconnected subgraphs whose disjoint components are the, to the physicist familiar, one-particle-irreducible (1PI) graphs. This construction was more recently extended also to the ``renormalisation" of IR divergences via the motic Hopf algebra \cite{Brown:2015fyf},\cite{Beekveldt:2020kzk}. In the following we discuss the general structure of such Hopf algebras.

Let us consider a vector space $\mathcal{H}$ over the field of rational numbers $\mathbb{Q}$ generated by the basis set of all non-isomorphic Feynman graphs in a given quantum field theory along with the empty graph $\emptyset$. Thus every element of $\mathcal{H}$ is a linear combination of Feynman graphs. In the following we shall often just talk about graphs but really have in mind Feynman graphs which are generated by the interactions of some Lagrangian. Note that since the elements of this Hopf algebra are non-isomorphic Feynman graphs, the labeling of the internal edges is  arbitrary and not kept track of. This means that if two graphs can be related by a relabelling of the internal edges, they are considered identical in $\mathcal{H}$. 
Instead the labeling of external edges is kept track of throughout. And graphs who are related by relabelling external lines are not considered isomorphic in $\mathcal{H}$.   

On this vector space we can now define a multiplication operation $m:\mathcal{H}\otimes\mathcal{H}\to\mathcal{H}$ defined as the disjoint union of graphs:
\bnon
m (\ga_1\otimes \ga_2) = \ga_1\sqcup\ga_2\,,
\enon
where $\ga_1$,$\ga_2$ are graphs in $\mathcal{H}$. The definition is then extended linearly to all elements in $\mathcal{H}$. Clearly the empty graph $\emptyset$ is the multiplicative identity here, so from now on we denote the empty graph with $\ide$.  Armed with this multiplication, the linear space $\mathcal{H}$ becomes an algebra. Notice that the multiplication operation and hence the algebra is commutative. Next we define a coproduct operation $\Delta:\mathcal{H}\to\mathcal{H}\otimes\mathcal{H}$ given by 
\begin{equation}
\label{eq:corpoduct}
\Delta(\ga) = \ga\otimes\ide + \ide\otimes\ga + \sum_{\substack{\gamma\subsetneq\ga, \gamma\neq\ide \\ \gamma\in\mathcal{H}}}\gamma\otimes\ga/\gamma\,,
\end{equation}
where the sum goes over all bridgeless subgraphs $\gamma$ of $\Gamma$, and is extended to products and linear combinations. For the sake of making $\mathcal{H}$ into a bialgebra, we define a unit operation $u:\mathbb{Q}\to\mathcal{H}$, given by $q\to q\ide$ and a counit $\bar{e}:\mathcal{H}\to \mathbb{Q}$ as, with $\ga$ a graph in $\mathcal{H}$:
\bnon
\bar{e}(\ga) = 
\begin{cases}
1\qquad \text{if $\ga = \ide$} \\
0\qquad \text{otherwise}
\end{cases}\,,
\enon
which is extended linearly. With these definitions $\mathcal{H}$ becomes a coalgebra because the properties of counitarity and coassociativity hold:
\begin{eqnarray}
\label{eq:counitarity}
(id\otimes \bar{e})\circ\Delta &\cong& (\bar{e}\otimes id)\circ\Delta \qquad\quad  \text{(counitarity)}\,,\\
\label{eq:coassociativity}
(id\otimes \Delta)\circ\Delta &=& (\Delta\otimes id)\circ\Delta\qquad \text{(coassociativity)}\,,
\end{eqnarray}
where $id$ is the identity operator that maps any element of $\mathcal{H}$ to itself. One can further show that the multiplication and coproduct are compatible making $\mathcal{H}$ into a bialgebra. 



The bialgebra thus obtained respects two natural gradings: in the number of loops, and in the number of edges.  For a graded (and connected) bialgebra like $\mathcal{H}$ there then exists a natural definition for the antipode $S:\mathcal{H}\to\mathcal{H}$ which satisfies 
\bnon
m\circ(id\otimes S)\circ\Delta = m\circ(S\otimes id)\circ\Delta = u\circ\bar{e}
\enon
for all elements of $\mathcal{H}$ given by
\begin{eqnarray}
\label{eq:antipode1}
S &=& u\circ\bar{e}-m\circ(S\otimes(id-\bar{e}))\circ\Delta \nonumber \\
&=& -\ga -\sum_{\substack{\gamma\subsetneq\ga, \gamma\neq\ide 
\\ \gamma\in\mathcal{H}}}S(\gamma)\ga/\gamma \,,
\end{eqnarray}
or, equivalently, as
\begin{eqnarray}
\label{eq:antipode2}
S &=& u\circ\bar{e}-m\circ((id-\bar{e})\otimes S)\circ\Delta\nonumber\\
&=& -\ga -\sum_{\substack{\gamma\subsetneq\ga, \gamma\neq\ide \\ \gamma\in\mathcal{H}}}\gamma S(\ga/\gamma).
\end{eqnarray}


Note that this formalism can be extended straight forwardly to the more involved motic Hopf algebra, simply by suitably extending the sets of subgraphs entering in the coproduct eq.~(\ref{eq:corpoduct}), and subsequently also in the antipode in eqs.~(\ref{eq:antipode1},\ref{eq:antipode2}). All other operations, that is $m, \bar{e}$ and $u$,  remain unaffected. In the next section we wll discuss how to extend this framework to the case of asymptotic subraphs which enter in the expansion by subgraph.

The forest formula for the BPHZ $R$-operation can be obtained by summing over the UV divergent bridgeless  graphs as:
\begin{equation}
\label{eq:BPHZ}
R(\ga)=\phi(\ga) + Z(\ga) + \sum_{\substack{\gamma\subsetneq\ga, \gamma\neq\ide \\ \gamma\in\mathcal{H}}}\phi\big(Z(\gamma)\ast\ga/\gamma\big)
\end{equation}
where $\phi:\mathcal{H}\to A$ maps a graph to its integrated expression, with $A$ the algebra of analytic functions in the external momenta and masses and dimension, and $\mathcal{H}$, the Connes-Kreimer Hopf algebra, defined in section \ref{sec:conneskreimer}. $Z:\mathcal{H}\to A$ is the UV counterterm map while the $\ast$ symbol inserts the counterterm expression into the relevant vertex. The counterterm $Z(\Gamma)$ is particularly simple if its SDD is logarithmic,  $\Omega=0$. Then $Z(\Gamma)$, defined by 
\begin{equation}
\label{eq:Zlog}
Z(\ga)=-K\Big(\phi(\ga) +\sum_{\substack{\gamma\subsetneq\ga, \gamma\neq\ide 
\\ \gamma\in\mathcal{H}}}Z(\gamma)\phi(\ga/\gamma) \Big)\,,    
\end{equation}
is just a number and the insertion product `$*$' reduces to the standard product `$\cdot$'. Note the similarity of the structure of $Z$,  eq.(\ref{eq:Zlog}), with the antipode in eq.(\ref{eq:antipode1}). This is precisely what motivated to think of $Z$ as a twisted antipode. For graphs which only have logarithmic divergences and subdivergences the $R$ operation then simplifies as follows.
\begin{equation}
\label{eq:BPHZlog}
R(\ga)=\phi(\ga) + Z(\ga) + \sum_{\substack{\gamma\subsetneq\ga, \gamma\neq\ide \\ \gamma\in\mathcal{H}}}Z(\gamma)\cdot\phi\big(\ga/\gamma\big)\,.
\end{equation}
For this case the Hopf algebra structure is most transparent. The maps $Z$ and $\phi$ are homomorphisms and thus belong to the set $\mathcal{G}=\mathrm{char}(\mathcal{H},A)$ of characters that map elements in the Hopf algebra $\mathcal{H}$ to the algebra $A$. 

It so happens that this set of characters $\mathcal{G}$ also forms a group under the convolution product, which for any two elements $g_1,g_2\in \mathcal{G}$ is defined as:
\bnu
\label{eq:convproduct}
g_1\star g_2 = m_A\circ(g_1\otimes g_2)\circ\Delta_{\mathcal{H}}\,,
\enu
where $m_A$ is the multiplication operation for the algebra $A$. It is straightforward to see that $\mathcal{G}$ is a group with the identity element given by $e=u_{\mathcal{H}}\circ\bar{e}_{\mathcal{H}}$ and the inverse of any element $g\in\mathcal{G}$ given by $g\circ S$ where $S$ is the antipode of $\mathcal{H}$, with the associativity of $\star$ guaranteed by the coassociativity of the coproduct. With this, we can compactly write the $R$-operation as a convolution product: 
\bnon
R = Z\star\phi = m_A\circ(Z\otimes \phi)\circ\Delta_{\mathcal{H}}\,.
\enon
As said before, this works perfectly well for logarithmic counterterms in purely scalar theories, but requires further justification for the case of higher degree divergences or more generally non-scalar interactions. A workaround in such cases is given by the
projector pairing formalism \cite{Ebrahimi-Fard:2005pga,Kreimer:2009iy,aihpd,Kissler:2019vga,Prinz:2019awo,Prinz:2022qll} in which every element of the Hopf algebra is promoted to a pair $(\ga,\sigma_i)$ where the $\sigma_i$'s project onto form factors that form a partition of unity. 

With the appropriate projectors we can then project out the interactions created by the $\ast$ symbol in terms of convolution products of characters, the sum of all the projections then gives us the total $R$-operation.

\section{The Hopf algebra of asymptotic subgraphs}
\label{hopfdef}
We start this section with a discussion of the graph-theoretic aspects of AI subgraphs, and show that they give rise to their own Hopf algebra which is closely related to the motic Hopf Algebra. 

\subsection{Motic and asymptotic subgraphs}
\label{Motic}
The criteria defining an AI subgraph are closely related to those defining a Euclidean IR subgraph, as they appear in the context of the $R^*$-operation. We start by reviewing the definition of IR subgraphs as given in \cite{Beekveldt:2020kzk}. Let $\gamma$ be a subgraph of a Feynman diagram $\Gamma$. Now let $\bar{\gamma}=\Gamma\backslash\gamma$ be the complement of $\gamma$ in $\Gamma$, that is obtained by deleting all edges and vertices of $\gamma$ in $\Gamma$. If the following conditions hold:
\blsd
\item $\bar{\gamma}$ contains all external lines and massive lines of $\Gamma$,
\item each connected component of the graph obtained by contracting all the massive lines in $\bar{\gamma}$ and identifying all the external lines in $\bar{\gamma}$ is 1PI,
\elsd
then $\gamma$ is an \emph{IR subgraph}, and we will refer to the corresponding $\bar \gamma$ as the \emph{IR cograph}. Note that we do \emph{not} require an IR cograph to lead to an IR divergence, for which a necessary condition would be that the contracted graph $\Gamma/\bar\gamma$ has a non-positive SDD. 

Let us assume the existence of a number of soft (small) scales (masses or momenta), around which we wish to Taylor expand. Now, let $\Gamma_0$ be the graph obtained by putting all these soft scales to zero in $\ga$. We will now show that any AI subgraph, $\gai$, of $\ga$ is in one-to-one correspondence with an IR cograph, $\bar{\gamma}$, of $\Gamma_0$. This follows since:
\blsd
\item if all soft momenta and masses are set to zero, the only external legs and massive lines in $\Gamma_0$ are the hard legs and heavy lines. Thus the first criterion for $\bar{\gamma}$ becomes identical to the first criterion for $\gai$,
\item the second criterion for $\bar{\gamma}$ is manifestly identical to the second criterion for $\gai$.
\elsd
Thus, we arrive at the following Lemma.
\begin{lem}\label{aitoir}
Let $\ga_0:=\tau(\ga)$ be the graph obtained from $\ga$ when all soft scales are set to zero through the operation $\tau$. Then there exists a bijection between each AI subgraph $\gai$ of $\ga$ and an IR cograph $\bar{\gamma}$ of $\ga_0$, such that:
\begin{equation*}
  \bar\gamma=\tau(\gai) 
\end{equation*}
\end{lem}

\noindent Let us explore this with a simple example. Consider
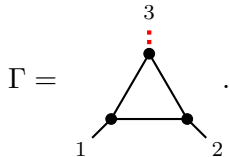
\begin{equation*}
 \ga =    \begin{minipage}{2.0cm}
\begin{tikzpicture}
\draw [thick] (0,0.866) -- (-0.5,0);
\draw [ultra thick, dotted, red] (0,0.866) -- (0,1.216);
\draw [thick] (0.5,0) -- (0,0.866);
\draw [thick] (0.5,0) -- (0.75,-0.25);
\draw [thick] (-0.5,0) -- (-0.75,-0.25);
\draw [thick] (-0.5,0) -- (0.5,0);
\filldraw (-0.5,0) circle (2pt);
\filldraw (0.5,0) circle (2pt);
\filldraw (0, 0.866) circle (2pt);
\draw (0,1.416) node{${\scriptstyle 3}$};
\draw (0.9,-0.4) node{${\scriptstyle 2}$};
\draw (-0.9,-0.4) node{${\scriptstyle 1}$};
\end{tikzpicture}
\end{minipage}\,.
\end{equation*}
Here, the dotted leg marked by 3 is the soft momentum and 1 and 2 are hard. All internal lines are massless. It is not hard to see that $\ga$ has two proper AI subgraphs:
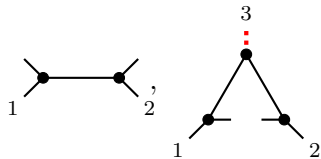
\begin{equation*}
\begin{minipage}{2.0cm}
\begin{tikzpicture}
\draw [thick] (0.5,0) -- (0.75,0.25);
\draw [thick] (-0.5,0) -- (-0.75,0.25);
\draw [thick] (0.5,0) -- (0.75,-0.25);
\draw [thick] (-0.5,0) -- (-0.75,-0.25);
\draw [thick] (-0.5,0) -- (0.5,0);
\filldraw (-0.5,0) circle (2pt);
\filldraw (0.5,0) circle (2pt);
\draw (0,0.4) node{  };
\draw (0.9,-0.4) node{${\scriptstyle 2}$};
\draw (-0.9,-0.4) node{${\scriptstyle 1}$};
\end{tikzpicture}
\end{minipage}, 
\begin{minipage}{2.0cm}
\begin{tikzpicture}
\draw [thick] (0,0.866) -- (-0.5,0);
\draw [ultra thick, dotted, red] (0,0.866) -- (0,1.216);
\draw [thick] (0.5,0) -- (0,0.866);
\draw [thick] (0.5,0) -- (0.75,-0.25);
\draw [thick] (-0.5,0) -- (-0.75,-0.25);
\draw [thick] (0.5,0) -- (0.2,0);
\draw [thick] (-0.5,0) -- (-0.2,0);
\filldraw (-0.5,0) circle (2pt);
\filldraw (0.5,0) circle (2pt);
\filldraw (0, 0.866) circle (2pt);
\draw (0,1.416) node{${\scriptstyle 3}$};
\draw (0.9,-0.4) node{${\scriptstyle 2}$};
\draw (-0.9,-0.4) node{${\scriptstyle 1}$};
\end{tikzpicture}
\end{minipage}\,.
\end{equation*}
Note that the second AI subgraph does not actually contribute , since the corresponding contracted graph $\Gamma/\gamma_{AI}$ is scaleless. Let us now set the momentum of leg 3 to zero to obtain
\begin{equation*}
\Gamma_0 =    \begin{minipage}{2.0cm}
\begin{tikzpicture}
\draw [thick] (0,0.866) -- (-0.5,0);
\draw [thick] (0.5,0) -- (0,0.866);
\draw [thick] (0.5,0) -- (0.75,-0.25);
\draw [thick] (-0.5,0) -- (-0.75,-0.25);
\draw [thick] (-0.5,0) -- (0.5,0);
\filldraw (-0.5,0) circle (2pt);
\filldraw (0.5,0) circle (2pt);
\filldraw (0, 0.866) circle (2pt);
\draw (0.9,-0.4) node{${\scriptstyle 2}$};
\draw (-0.9,-0.4) node{${\scriptstyle 1}$};
\end{tikzpicture}
\end{minipage}\,,
\end{equation*}
with the IR cographs of $\Gamma_0$ given by
\begin{equation*}
\begin{minipage}{2.0cm}
\begin{tikzpicture}
\draw [thick] (0.5,0) -- (0.75,0.25);
\draw [thick] (-0.5,0) -- (-0.75,0.25);
\draw [thick] (0.5,0) -- (0.75,-0.25);
\draw [thick] (-0.5,0) -- (-0.75,-0.25);
\draw [thick] (-0.5,0) -- (0.5,0);
\filldraw (-0.5,0) circle (2pt);
\filldraw (0.5,0) circle (2pt);
\draw (0,0.4) node{  };
\draw (0.9,-0.4) node{${\scriptstyle 2}$};
\draw (-0.9,-0.4) node{${\scriptstyle 1}$};
\end{tikzpicture}
\end{minipage},
\begin{minipage}{2.0cm}
\begin{tikzpicture}
\draw [thick] (0,0.866) -- (-0.5,0);
\draw [thick] (0.5,0) -- (0,0.866);
\draw [thick] (0.5,0) -- (0.75,-0.25);
\draw [thick] (-0.5,0) -- (-0.75,-0.25);
\draw [thick] (0.5,0) -- (0.2,0);
\draw [thick] (-0.5,0) -- (-0.2,0);
\filldraw (-0.5,0) circle (2pt);
\filldraw (0.5,0) circle (2pt);
\filldraw (0, 0.866) circle (2pt);
\draw (0.9,-0.4) node{${\scriptstyle 2}$};
\draw (-0.9,-0.4) node{${\scriptstyle 1}$};
\end{tikzpicture}
\end{minipage}\,,
\end{equation*}
which are in one-to-one correspondence to the $\gai$'s of $\Gamma$ with leg $3$ soft.

Let us now turn to the definition of a motic mass-momentum spanning (MM) subgraph. For a subgraph $\gamma$ of $G$ to be \emph{mass-mommentum spanning} in $G$  it must contain all the masses and external momenta of the parent graph $G$, with the external momenta all in a single connected component in $\gamma$. We can then state the definition for a subgraph to be motic. For a subgraph $\gamma$ of $G$ to be motic, each proper   subgraph $\mu$ of $\gamma$ must be MM in $\gamma$ and have strictly less loops than $\gamma$. Given the recursivity of this definition it may be difficult to grasp its meaning without examples, however the following theorem by Beekveldt, one of the authors and Borinsky \cite{Beekveldt:2020kzk} provides an alternative definition for the class of MM subgraphs. 
\begin{theo}\label{irtomm}
A subgraph $\gamma\subseteq\ga$ is a motic mass-momentum spanning subgraph of $\ga$ if and only if it is an IR cograph of $\ga$.
\end{theo}
Due to the bijection between AI subgraphs and IR subgraphs, it is thus natural to expect that there will be a corresponding theorem connecting AI subgraphs with certain kinds of motic subgraphs which span only the hard momenta and heavy masses. To make this more precise we define a subgraph $\gamma$ of $\ga$ to be \emph{hard mass-momentum spanning} if the following conditions are met:
\blsd
\item  the subgraph contains all the heavy lines of $\ga$,
\item  one connected component has all the hard legs.
\elsd
Then, a subgraph $\gamma$ is said to be \textit{motic hard mass-momentum spanning} (MHM) in $\ga$ if:
\blsd
\item $\gamma$ is a hard mass-momentum spanning subgraph of $\ga$, 
\item every proper subgraph $\mu$ of $\gamma$ which is MHM\ in $\gamma$ has at least one loop less than $\gamma$.
\elsd
Notice that the MHM property is inheritive, that is, any graph $\mu\subseteq\gamma$ that is MHM\ in $\gamma$ is also MHM\ in $\ga$ iff $\gamma$ is a MHM\ subgraph of $\ga$. 
Using the terminology introduced in Lemma \ref{aitoir} we then arrive at the following Lemma. 
\begin{lem}\label{mmtohm}
The MHM\ subgraphs of $\ga$ are in one-to-one correspondence with the MM\ subgraphs of $\ga_0=\tau(\ga)$.
\end{lem}
In order to check whether a subgraph is MHM\ or not one does not actually need to check the second criterion for all proper subgraphs $\mu$ of $\gamma$.  Instead it is sufficient to check that all MHM\ subgraphs $\mu$, obtained by deleting a single edge in $\gamma$, have one loop less than $\gamma$. Brown showed this to be true for MM subgraphs in ref.~\cite{Brown:2015fyf}. With Lemma \ref{mmtohm} this property also holds for MHM subgraphs. 

With lemmas \ref{aitoir} and \ref{mmtohm}  and theorem \ref{irtomm}, we are then in a position to state the following theorem.
\begin{theo}\label{aitohm}
A subgraph $\gamma\subseteq\ga$ is a motic hard mass-momentum spanning subgraph of $\ga$ if and only if it is an AI subgraph of $\ga$.
\end{theo}
Theorem \ref{aitohm} also implies the following corollary which was proved by Brown for MM subgraphs in ref.~\cite{Brown:2015fyf}.
\begin{cor}\label{inherit}
Let $\mu\subset\gamma\subset\ga$ be subgraphs. Then
\begin{enumerate}
    \item $\mu$ is an AI subgraph of $\gamma$ and $\gamma$ is an AI subgraph of $\ga \Longleftrightarrow \mu $ is an AI subgraph of $\ga$
    \item $\gamma$ is an AI subgraph of $\ga \Longleftrightarrow \gamma/\mu$ is an AI subgraph of $\ga/\mu$
\end{enumerate}
\end{cor}
MM subgraphs inherit all the masses and momenta of the parent graph\cite{Brown:2015fyf}. Similarly, MHM subgraphs inherit all the hard masses and hard momenta of the parent graph. Thus, the associated contracted graph depends only on soft scales. We further define a \emph{primitive} MHM subgraph to be one which does not have any proper MHM subgraphs.

\subsection{Hopf algebra construction} 

Let $\Gamma^{[1]}$ be the set of lines of a 1PI-graph $\ga$ (both external and internal). Let $P$ be the set of lines corresponding to the soft momenta and soft masses. Then, $Q\equiv \Gamma^{[1]}\backslash P$ is the set of hard momentum legs and heavy mass lines. For there to be a meaningful (asymptotic) expansion we impose that $Q$ must be such that $\ga$ has kinematic dependence on the hard scales. 
In the following we will use that:
\begin{itemize}
\item $\gamma_{AI}$ refers to an AI subgraph of $\ga$ which, by theorem \ref{aitohm}, is also motic hard mass-momentum spanning in $\ga$ w.r.t the scales in $Q$.
\item ${\ga}/ {\gamma_{AI}}$ is obtained from $\ga$ by contracting all the \textbf{internal edges} in $\gamma_{AI}$ to a point.
\end{itemize}
Using the definitions as given above, we define a set $H_\ga$ which initially contains all AI subraphs $\gamma_{AI}$ of $\ga$ and the corresponding ${\ga}/ {\gamma_{AI}}$. We now perform the following steps: 
\begin{enumerate}
    \item For each $\gamma\in H_\ga$ add to $H_\ga$ all possible AI subraphs $\gamma_{AI}$ of $\gamma$ and the corresponding ${\gamma}/ {\gamma_{AI}}$.
    \item Repeat step 1 until no further proper subgraphs can be obtained.
    \item Identify all the single vertex graphs with the empty graph.
\end{enumerate}
The grading in terms of loops ensures that the above recursion terminates, so the set constructed as above is well-defined, and finite. We can extend $H_\ga$ to the set $H:=\bigcup H_\ga$, where the union is over all graphs $\ga$ with the same set of hard scales $Q$. The $\mathbb{K}$-vector space $\hopf$ spanned by all possible disjoint unions of the elements of $H$, with $\mathbb{K}$ any suitable number field, for us most importantly the rational numbers $\mathbb{Q}$, then becomes an algebra with
\begin{itemize}
    \item multiplication $m: \hopf \otimes \hopf \to \hopf$ given by concatenation of graphs,
    \item unit $u: \mathbb{K}\to \hopf$ given by $q \mapsto q\mathbb{I}$
\end{itemize} and extended linearly. Here $\mathbb{I}$ is the multiplicative identity given by the empty graph. 

\noindent Now we need to find a compatible counit and coproduct on $\hopf$ to promote it to a bialgebra. The counit $\bar{e}: \hopf \to \mathbb{K}$ is given by
\bnon
\bar{e}(\gamma) =  
\begin{cases}
1 & \text{if } \gamma = \mathbb{I} \\
0 & \text{otherwise}
\end{cases}
\enon
and extended linearly. Following similar arguments as in \cite{Brown:2015fyf} this counit is compatible with $m$ and $u$. Let us now define the coproduct $\Delta :\hopf \to \hopf \otimes \hopf$ for the elements in the generator set $H$ of $\hopf$ as:
\bnu
\label{eq:coproducta}
\Delta(\gamma) = 
\begin{cases}
\ide\otimes\gamma + \gamma\otimes\ide +  \displaystyle{\sum_{\substack{\mu_{AI}\subsetneq\gamma \\ \mu_{AI}\neq \ide}}{\mu_{AI}\otimes\gamma/\mu_{AI}}} & \text{if } \gamma\in H \text{ is non-empty} \\
\ide\otimes\ide & \text{if } \gamma = \ide
\end{cases}
\enu
If $\gamma$ is the disjoint union of more than one $\delta_i$ that are elements of $H$,
\bnu
\label{eq:coproductb}
\Delta(\gamma) = \Delta(\sqcup_i\delta_i) = \sum \left(\sqcup_i\delta_i'\right)\otimes\left(\sqcup_i\delta_i''\right)
\enu
where we used Sweedler's notation. The $\sqcup_i\delta_i'$ run over all possible disjoint unions of subgraphs of each $\delta_i$ that appear on the left side of the $\otimes$ symbol of the coproduct expression for $\delta_i$, likewise for $\sqcup_i\delta_i''$.

\subsection{Proof of Hopf Algebraic Structure}\label{prooo}
\label{sec:hopfproof}

In the following we show that the  operations $(u,m,\bar e,\Delta)$ indeed fulfill all the requirements of a bialgebra, which together with an antipode $S$, to be defined below, is then promoted to a Hopf algebra.  While it is clear that $\hopf$ is a unital associative algebra with $(u,m)$, we still need to show that counitarity, coassociativity, stated in eqs. (\ref{eq:counitarity}) and (\ref{eq:coassociativity}), 
and the compatibility of $m$ and $\Delta$ hold. 

\begin{pro}
\label{pro:counitarity}
The coproduct $\Delta$ fulfills counitarity with $\bar{e}$.
\end{pro}
\begin{proof}
It is straight forward to see that for any non-empty graph $\gamma$,
\bnon
m\circ(\bar{e}\otimes id)\circ \Delta(\gamma) = \gamma = m\circ(id \otimes \bar{e})\circ \Delta(\gamma)
\enon
and for the empty graph,
\bnon
m\circ(\bar{e}\otimes id)\circ \Delta(\mathbb{I}) = \mathbb{I} = m\circ(id \otimes \bar{e})\circ \Delta(\ide) 
\enon
Thus we have counitarity.
\end{proof}
\vspace{1pt}
\begin{pro}\label{coass}
The coproduct $\Delta$ is coassociative.
\end{pro}
\begin{proof}
We will be using the following simplification of notation for brevity, in the case of a non-empty graph $\Phi \in H$ :
\bnon
\sum_{\gamma\underset{AI}{\subseteq}\Phi}{\gamma\otimes\Phi/\gamma}\equiv\Phi\otimes\ide+{\sum_{\substack{\gamma_{AI}\subsetneq\Phi \\ \gamma_{AI}\neq \ide}}{\gamma_{AI}\otimes\Phi/\gamma_{AI}}}
\enon
Notice that for this equation to hold we require $\gamma/\gamma = \ide$. We therefore implicitely identify all single vertex graphs with the identity $\ide$, as was done in the definition of $H$. For the remainder of this proof we will assume $\subseteq$ to mean $\underset{AI}{\subseteq}$ unless otherwise specified. 
For a non-empty $\Phi$ we have:
\begin{eqnarray}
&&(id\otimes \Delta)\circ \Delta(\Phi)  =  (id\otimes \Delta)\circ(\ide\otimes\Phi + \displaystyle{\sum_{\gamma\subseteq\Phi}{\gamma\otimes\Phi/\gamma}}) \nonumber\\
&&\qquad = \ide\otimes\Big(\ide\otimes\Phi + \displaystyle{\sum_{\gamma\subseteq\Phi}{\gamma\otimes\Phi/\gamma}}\Big) + \sum_{\gamma\subsetneq\Phi}\gamma\otimes\Big(\ide\otimes\Phi/\gamma + \sum_{\substack{\eta\subseteq\Phi/\gamma \\ \Phi/\gamma\ne \ide}}\eta\otimes\Phi/\gamma/\eta\Big) 
+  \Phi\otimes\ide\otimes\ide \nonumber\\
&&\qquad= \ide\otimes\ide\otimes\Phi + \displaystyle{\sum_{\gamma\subseteq\Phi}{\ide\otimes\gamma\otimes\Phi/\gamma}} + \Big(\sum_{\gamma\subsetneq\Phi}\gamma\otimes\ide\otimes\Phi/\gamma + \Phi\otimes\ide\otimes\ide\Big) \nonumber\\
& & 
\qquad \qquad + \sum_{\gamma\subsetneq\Phi}\sum_{\eta\subseteq\Phi/\gamma}\gamma\otimes\eta\otimes\Phi/\gamma/\eta \nonumber\\[5pt]
&& \qquad =  \ide\otimes\ide\otimes\Phi + \displaystyle{\sum_{\gamma\subseteq\Phi}{\ide\otimes\gamma\otimes\Phi/\gamma}} + \sum_{\gamma\subseteq\Phi}\gamma\otimes\ide\otimes\Phi/\gamma +\sum_{\gamma\subsetneq\Phi}\gamma\otimes\Phi/\gamma\otimes\ide \nonumber\\
& & 
\label{eq:coass1}
\qquad\qquad
+ \sum_{\gamma\subsetneq\Phi}\sum_{\eta\subsetneq\Phi/\gamma}\gamma\otimes\eta\otimes\Phi/\gamma/\eta
\end{eqnarray}
On the other hand we have
\begin{eqnarray}
(\Delta\otimes id)\circ \Delta(\Phi) & = & (\Delta\otimes id)\circ(\ide\otimes\Phi + \displaystyle{\sum_{\gamma\subseteq\Phi}{\gamma\otimes\Phi/\gamma}}) \nonumber\\
&=& \ide\otimes\ide\otimes\Phi + \sum_{\gamma\subseteq\Phi}(\ide\otimes\gamma + \sum_{\mu\subseteq\gamma}\mu\otimes\gamma/\mu) \otimes\Phi/\gamma\nonumber\\
&=& \ide\otimes\ide\otimes\Phi + \sum_{\gamma\subseteq\Phi}\ide\otimes\gamma\otimes\Phi/\gamma + \sum_{\gamma\subseteq\Phi}\gamma\otimes\ide\otimes\Phi/\gamma \nonumber\\
& & 
\label{eq:coass2}
+\sum_{\gamma\subsetneq\Phi}\gamma\otimes\Phi/\gamma\otimes\ide 
+\sum_{\gamma\subsetneq\Phi}\sum_{\mu\subsetneq\gamma}\mu\otimes\gamma/\mu\otimes\Phi/\gamma 
\end{eqnarray}
The first four terms in eqs.~(\ref{eq:coass1}) and (\ref{eq:coass2}) all contain the empty graph dependent terms and are clearly equal. The fifth terms refer to the part of the expressions where no component in the double tensor product is an empty graph. Now we let $\gamma/\mu \equiv \eta$, then $\Phi/\gamma= \Phi/\mu/\eta$, since  $(\Phi/\mu)/(\gamma/\mu) = \Phi/\gamma$. The latter is a fairly standard result from graph theory, also used, e.g., in ref.~\cite{manchon2014bialgebra}. Thus the fifth sum in eq.~\eqref{eq:coass2} can be rewritten as
\bnon
\sum_{\mu}\sum_{\eta}\mu\otimes\eta\otimes\Phi/\mu/\eta \,.
\enon
This has the same structure as the fifth term in the first case, but we need to make sure that the sum runs over the same values of the dummy variables $\mu$ and $\eta$. As $\gamma$ ran over all the proper non-empty AI subgraphs, while in turn $\mu$ ran over the proper non-empty AI subgraphs of $\gamma$, by corollary \ref{inherit} every $\mu$ is a proper non-empty AI subgraph of $\Phi$, and so the sum runs over $\mu\underset{AI}{\subsetneq}\Phi$. \\

By corollary \ref{inherit}, $\gamma/\mu$ is an AI subgraph of $\Phi/\mu$ for every $\mu$. Thus by $\gamma$ running over all non-empty proper AI subgraphs of $\Phi$, $\eta = \gamma/\mu$ runs over all non-empty proper AI subgraphs of $\Phi/\mu$. So $\eta\underset{AI}{\subsetneq}\Phi/\mu$.\\

\noindent Thus the fifth terms in the two expressions are equal. For multiple connected components we can use the exact same arguments, and for the empty graph the coassociativity is trivial. Hence, the coproduct is coassociative. 
\end{proof}\vspace{8pt}

\begin{pro}
\label{pro:compatibility}
The coproduct $\Delta$ and the product $m$ are compatible.
\end{pro}
\begin{proof}
For our purposes here we write the coproduct in Sweedler's notation
\bnon
\Delta(\gamma) = \sum \gamma'\otimes\gamma''\,,
\enon
where the summation is understood to be over the relevant terms in the coproduct.
Let $\Phi_1$ and $\Phi_2$ be any two graphs in $\hopf$. Then,
\bnon
\Delta\circ m(\Phi_1\otimes\Phi_2) = \Delta(\Phi_1\sqcup\Phi_2) = \sum (\Phi_1'\sqcup\Phi_2')\otimes(\Phi_1''\sqcup\Phi_2'')\,.
\enon
On the other hand, with $m_{\hopf\otimes\hopf}$ given by : 
\bnon
m_{\hopf\otimes\hopf} : \hopf\otimes\hopf\otimes\hopf\otimes\hopf \to \hopf\otimes\hopf \,,
\enon
\bnon
\gamma_1\otimes\gamma_2\otimes\gamma_3\otimes\gamma_4 \,,\longmapsto (\gamma_1\sqcup\gamma_2)\otimes(\gamma_3\sqcup\gamma_4)\,,
\enon
and $\Delta_{\hopf\otimes\hopf}$ given by
\bnon
\Delta_{\hopf\otimes\hopf}: \hopf\otimes\hopf \to \hopf\otimes\hopf\otimes\hopf\otimes\hopf \,, 
\enon
\bnon
\gamma_1\otimes\gamma_2 \longmapsto \sum \gamma_1'\otimes\gamma_2'\otimes\gamma_1''\otimes\gamma_2''\,.
\enon
For compatibility we need,
\bnon
m_{\hopf\otimes\hopf}\circ\Delta_{\hopf\otimes\hopf} = \Delta\circ m \,.
\enon
And indeed we have,
\bqa
m_{\hopf\otimes\hopf}\circ\Delta_{\hopf\otimes\hopf}(\Phi_1\otimes\Phi_2) &=& m_{\hopf\otimes\hopf}\Big(\sum \Phi_1'\otimes\Phi_2'\otimes\Phi_1''\otimes\Phi_2''\Big) \\[2pt]
&=& \sum (\Phi_1'\sqcup\Phi_2')\otimes(\Phi_1''\sqcup\Phi_2'') \\[4pt]
&=& \Delta\circ m(\Phi_1\otimes\Phi_2)\,.
\eqa
\end{proof}
\begin{pro}
$\hopf$ is a Hopf algebra.
\end{pro}
\begin{proof}
We note that propositions \ref{pro:counitarity} and \ref{coass} together imply that $\hopf$ is a co-algebra. Together with proposition \ref{pro:compatibility} it then follows that $\hopf$ is a bialgebra.

With the antipode $S: \hopf \to\hopf$, defined recursively via 
\bnon
S = u\circ\bar{e}-m\circ(S\otimes(id-\bar{e}))\circ\Delta = u\circ\bar{e}-m\circ((id-\bar{e})\otimes S)\circ\Delta\,,
\enon
which due to coassociativity satisfies for all elements of $\hopf$,
\bnon
m\circ(id\otimes S)\circ\Delta = m\circ(S\otimes id)\circ\Delta = u\circ\bar{e}\,,
\enon
$\hopf$ is then promoted to a Hopf algebra. As $\hopf$ is graded, the recursion is bound to terminate at $S(\ide) = \ide$ and thus the above statements are well-defined.
\end{proof}

\subsection{Verifying the propositions for some basic examples}\label{ex}
In this section we will give a few examples to demonstrate  coassociativity and the validity of the antipode we defined. We begin with a simple example, let us take $\hopf$ where:
\begin{equation*}
 \ga =    \begin{minipage}{2.0cm}
\begin{tikzpicture}
\draw [thick] (0,0.866) -- (-0.5,0);
\draw [ultra thick, red, dotted] (0,0.866) -- (0,1.216);
\draw [thick] (0.5,0) -- (0,0.866);
\draw [thick] (0.5,0) -- (0.75,-0.25);
\draw [thick] (-0.5,0) -- (-0.75,-0.25);
\draw [thick] (-0.5,0) -- (0.5,0);
\filldraw (-0.5,0) circle (2pt);
\filldraw (0.5,0) circle (2pt);
\filldraw (0, 0.866) circle (2pt);
\draw (0,1.416) node{${\scriptstyle 3}$};
\draw (0.9,-0.4) node{${\scriptstyle 2}$};
\draw (-0.9,-0.4) node{${\scriptstyle 1}$};
\end{tikzpicture}
\end{minipage}, \qquad P = \{3\}, Q = \{1,2\}
\end{equation*}
For this section  we use the red dotted lines to denote the soft legs. Now, the coproduct acted on $\ga$ gives 
\bnon
\Delta(\ga) = \ide\otimes\begin{minipage}{1.6cm}
\begin{tikzpicture}
\draw [thick] (0,0.866) -- (-0.5,0);
\draw [ultra thick, red, dotted] (0,0.866) -- (0,1.216);
\draw [thick] (0.5,0) -- (0,0.866);
\draw [thick] (0.5,0) -- (0.75,-0.25);
\draw [thick] (-0.5,0) -- (-0.75,-0.25);
\draw [thick] (-0.5,0) -- (0.5,0);
\filldraw (-0.5,0) circle (2pt);
\filldraw (0.5,0) circle (2pt);
\filldraw (0, 0.866) circle (2pt);
\end{tikzpicture}
\end{minipage} + 
\begin{minipage}{1.6cm}
\begin{tikzpicture}
\draw [thick] (0,0.866) -- (-0.5,0);
\draw [ultra thick, red, dotted] (0,0.866) -- (0,1.216);
\draw [thick] (0.5,0) -- (0,0.866);
\draw [thick] (0.5,0) -- (0.75,-0.25);
\draw [thick] (-0.5,0) -- (-0.75,-0.25);
\draw [thick] (-0.5,0) -- (0.5,0);
\filldraw (-0.5,0) circle (2pt);
\filldraw (0.5,0) circle (2pt);
\filldraw (0, 0.866) circle (2pt);
\end{tikzpicture}
\end{minipage}\otimes\ide \:\: + \,
\begin{minipage}{2.0cm}
\begin{tikzpicture}
\draw (0,0) -- (-1,0);
\draw (0,0) -- (0.35,-0.35);
\draw (-1,0) -- (-1.35,-0.35);
\draw [ultra thick, red, dotted](-1,0) -- (-1.35,0.35);
\draw [ultra thick, red, dotted](0,0) -- (0.35,0.35);
\filldraw (-1,0) circle (2pt);
\filldraw (0,0) circle (2pt);
\end{tikzpicture}
\end{minipage}\otimes
\begin{minipage}{2.0cm}
\begin{tikzpicture}
\draw[ thick, -] (0,0.5) arc (90:450:0.5);
\draw [ultra thick, red, dotted] (-1,0) -- (-0.5,0);
\draw (0.5,0) -- (0.85,0.35);
\draw (0.5,0) -- (0.85,-0.35);
\filldraw (-0.5,0) circle (2pt);
\filldraw (0.5,0) circle (2pt);
\end{tikzpicture}
\end{minipage}\;.
\enon
Where we did not consider any terms with tadpole graphs as they just add a zero contribution to the EBS expression. This is possible because if we remove all the tadpole graphs in $\hopf$, the resulting vector space is still a Hopf algebra, because in the coassociativity, compatibility and antipode tests the terms with tadpoles get removed from both sides of the relevant equation. So in this section we will not be including tadpole terms for the sake of brevity.\\[5pt]
Let us check if coassociativity holds for $\ga$:
\bqa
&&(id\otimes\Delta)\circ\Delta(\ga) = \ide\otimes\left(\ide\otimes\tria + \tria\otimes\ide + \stra\otimes\triastra\right) \\[3pt] 
&&\qquad + \stra\otimes\left(\ide\otimes\triastra + \triastra\otimes\ide\right) + \tria\otimes \ide\otimes\ide \\[12pt]
&&\qquad= \ide\otimes\ide\otimes\tria + \ide\otimes\tria\otimes\ide + \tria\otimes\ide\otimes\ide \\[3pt]
&&\qquad  + \ide\otimes\stra\otimes\triastra + \stra\otimes\ide\otimes\triastra \\[2pt]
&&\qquad +  \stra\otimes\triastra\otimes\ide
\eqa
\bqa
&&(\Delta\otimes id)\circ\Delta(\ga)  =  \left(\ide\otimes\tria + \tria\otimes\ide + \stra\otimes\triastra\right)\otimes\ide \\[2pt]
&&\qquad + \left(\ide\otimes\stra + \stra\otimes\ide\right)\otimes\triastra  + \ide\otimes\ide\otimes\tria \\[12pt]
&&\qquad =  \ide\otimes\ide\otimes\tria +\: \ide\otimes\tria\otimes\ide \:+ \tria\otimes\ide\otimes\ide \\[3pt]
&&\qquad  +\: \ide\otimes\stra\otimes\triastra \:+\: \stra\otimes\ide\otimes\triastra \\[2pt]
&&\qquad +\:  \stra\otimes\triastra\otimes\ide
\eqa
which are equal, and thus coassociativity holds. Let us check the left antipode:
\bqa
S \: = \;-m\circ(S\otimes(id-\bar{e}))\circ\Delta(\ga) &=& -\tria - S\left(\;\stra\right)\triastra\\[2pt]
&=& -\tria \:+\:\stra\triastra 
\eqa
while  the right antipode
\bqa
S \: = \;-m\circ((id-\bar{e})\otimes S)\circ\Delta(\ga) &=& -\tria - \stra S\left(\;\triastra\right)\\[2pt]
&=& -\tria \:+\:\stra\triastra \; .
\eqa
Thus confirming the equivalence of the two definitions of $S$.\\

We note that $S(\ga)$ also requires coassociativity to hold for $\ga$. Since the definition of $H_\ga$ incorporates the antipode the antipode test for the parent graph also confirms compatibility between coassociativity and antipode for all graphs in the generator set $H_\ga$. Therefore, for the next few examples we only show the antipode test for the parent graphs in different $\hopf$s.\\[10pt]
For the next example let us consider
\bnon
\ga = \logthn, \qquad P = \{2,4\}, Q = \{1,3\}
\enon
We have 
\bnon
\Delta(\ga) = \ide\;\otimes\logth + \logth\otimes\ide \: +\: \logthup\otimes\;\coup + \logthdown\otimes\;\codown \; .
\enon
Left Antipode:
\bqa
S(\ga) &=& -\logth -S\left(\logthup\right)\,\coup - S\left(\logthdown\right)\,\codown \\[5pt]
&=& -\logth + \logthup\coup + \logthdown\codown \; .
\eqa
Right antipode
\bqa
S(\ga) &=& -\logth -\logthup\, S\left(\,\coup\right) - \logthdown\, S\left(\,\codown\right) \\[5pt]
&=& -\;\logth + \logthup\:\coup + \logthdown\:\codown \; .
\eqa
which are again equal.

\noindent Let us now consider 
\bnon
\ga = \dtrian, \qquad P=\{1\}, Q=\{2,3\}
\enon
The coproduct is given by:
\bnu \label{coproduct:dtriangle}
\Delta(\ga) = \ide\otimes\dtria + \dtria\otimes\ide \; + \; \dstra\otimes\;\codstra + \dtstra\otimes\;\codtstra \; .
\enu
The antipode test yields:
\bqa
S(\ga) &=& -\dtria -S\left(\dstra\:\,\right)\codstra - S\left(\dtstra\:\,\right)\codtstra \\[7pt]
&=& -\dtria -\left(-\dstra\:\,\right)\codstra - \left(-\dtstra - S\left(\dstra\:\,\right)\codstras \right)\codtstra\\[7pt]
&=& -\dtria + \dstra\:\codstra + \dtstra\:\codtstra -\dstra\:\codstras\codtstra \\[2pt]
\eqa

\bqa
S(\ga) &=& -\dtria -\dstra S\left(\:\:\codstra\right) -\dtstra S\left(\:\:\codtstra\right) \\[7pt]
&=& -\dtria - \dstra\left(-\;\codstra-\codstras S\left(\:\:\codtstra\right)\right) -\dtstra \left(-\;\codtstra\right)\\[7pt]
&=& -\dtria + \dstra\:\codstra + \dtstra\:\codtstra -\dstra\:\codstras\codtstra\\
\eqa
Again the left and right antipodes are equal. \\[5pt]
Till now we only discussed examples with only large momenta and no large masses. In the following we consider a case where we have both:
\bnon
\ga = \logm\,,
\enon
where the thick internal lines carry large mass. We thus have 
\bnon
\Delta(\ga) = \ide\,\otimes\logm + \logm\otimes\,\ide \; + \logmup\otimes\logmcoup + \logmdown\otimes\logmcodown\; .
\enon
It is worth noticing here that in the case of heavy lines we can also have disconnected AI subgraphs. Now, we have the left antipode as:
\bqa
S(\ga) &=& -\logm -S\left(\logmup\right)\,\logmcoup - S\left(\logmdown\right)\,\logmcodown \\[5pt]
&=& -\logm + \logmup\logmcoup + \logmdown\logmcodown
\eqa
while the right antipode is:
\bqa
S(\ga) &=& -\logm -\logmup\, S\left(\,\logmcoup\right) - \logmdown\, S\left(\,\logmcodown\right) \\[5pt]
&=& -\;\logm + \logmup\:\logmcoup + \logmdown\:\logmcodown\;.
\eqa
The right and left antipodes are again equal, as expected.

\section{Hopf algebraic formulation of the expansion by subgraph}\label{hopfformulation}

In this section we shall formulate the expansion by subgraph in terms of the asymptotic Hopf algebra defined in section \ref{hopfdef}. We will start the discussion by focusing on the leading/logarithmic term in the expansion only, in this case no higher order derivatives are required for the leading approximation in $\lambda\to 0$. From the Hopf algebraic perspective this is  simpler since the integrals over subgraphs and contracted graphs factorize exactly. 

The Hopf structure of higher-order expansions is discussed  in section \ref{higherorder} and builds on the factorisation of the integrand. 

\subsection{The logarithmic case}

The logarithmic case applies to the leading term of the expansion for a special class of diagrams, such that the following criteria hold:
\begin{itemize}
    \item the parent graph $\ga$ and all its AI subgraphs $\gai$ have logarithmic SDD,
    \item the order of expansion is $a=0$.
\end{itemize}
Diagrams which are MHM (w.r.t.\ a given set of hard scales), and satisfy the first criterion, will be referred to as \emph{subgraph log-divergent} (SLD).

Following eq.\ (\ref{ebrnew}), given an SLD diagram $\Gamma$, the leading term in the expansion by subgraph is given by 
\bnu
\label{eq:LSDexp}
    \widetilde{T}^0(\ga) = \sum_{\gai\subseteq\ga}T^{0}(\gai)\,\widetilde{T}^0(\ga/\gai)\,.
\enu
It will now turn out convenient to introduce another Taylor-like operator $\overline{T}$ which is defined as follows:
\begin{equation}
\label{eq:Tbar}
    \overline{T}^{(n)}(\gamma)=
\begin{cases}
    \quad\ide  & \text{if $\gamma=\ide$ }, \\
    -T^{(n)}(\gamma) &  \text{if $\gamma=\gamma_{AI}$ or $\gamma=\Gamma/\gamma_{AI}$} \,.
\end{cases}  
\end{equation}
These are the cases which are actually required in our construction\footnote{More generally though one could consider acting $\overline{T}$ on products of AI subgraphs. For such cases, to ensure that $\overline{T}$ is a homomorphism one would require that 
$\overline{T}^{(n)}(\gamma)=(-1)^C T^{(n)}(\gamma)$ where $C$ is the number of elements in $H$, the set of generators of Hopf algebra, that $\gamma$ is a disjoint union of.}. We now give a theorem which reformulates eq.\ (\ref{eq:LSDexp}) in a Hopf algebraic language.
\begin{theo}
\label{thm:logAHopfalgebra}
The leading term in the expansion by subgraph for an SLD diagram is determined by
\bnu
m_{A}\circ(\overline{T}^{0}\otimes \widetilde{T}^0)\circ\; \Delta (\ga) = 0
\enu
\end{theo}

\begin{proof}
The proof is now straight forward:
\begin{align*}
m_{A}\circ(\overline{T}^{0} & \otimes \widetilde{T}^0)\circ\; \Delta (\ga) \\
&= m_{A}\circ(\overline{T}^{0}\otimes \widetilde{T}^0)\circ
\Big( \ide \otimes \Gamma + \sum_{\gai\subseteq\Gamma}  \gai\otimes \Gamma/\gai    \Big)\\
&= m_{A}\circ \Big( \overline{T}^{0}(\ide) \otimes \widetilde{T}^0(\Gamma) + \sum_{\gai\subseteq\Gamma}   \overline{T}^{0}(\gai)\otimes \widetilde{T}^0(\Gamma/\gai)    \Big)\\
&=\widetilde{T}^0(\ga) - \sum_{\gai\subseteq\ga}T^{0}(\gai)\,\widetilde{T}^0(\ga/\gai)\\
&=0\,
\end{align*}
where we used the definition of the coproduct, eq.  (\ref{eq:coproducta}), to get to the second line, and eq. (\ref{eq:LSDexp}) to get to the last line. Note that here $m_A$ is the product operation defined on the algebra of the integrals, $A$, defined in section \ref{bg}; see also the discussion below eq.\ (\ref{eq:convproduct}).
\end{proof}


Theorem \ref{thm:logAHopfalgebra} is one of the central results of this work - as it neatly expresses the expansion by subgraph in terms of the product and coproduct of the underlying Hopf algebra. However, in complete analogy to the case of renormalisation Hopf algebras, this result can also be written more compactly using the convolution product defined in eq.\ (\ref{eq:convproduct}), here now in the context of the asymptotic Hopf algebra, as follows  
\bnu
\label{loga}
    \overline{T}^0\star\widetilde{T}^0\;(\Gamma) = 0\,,
\enu
where $\Gamma$ is again SLD. Let us now consider some examples. We start with the one-loop triangle:
\begin{align*}
&\overline{T}^0\star\widetilde{T}^0\;\Big(\tria\Big) =m_{A}\circ(\overline{T}^{0}\otimes \widetilde{T}^0)\circ\; \Delta \Big(\tria\Big)\\[3pt]
& = m_{A}\circ(\overline{T}^{0}\otimes \widetilde{T}^0)\circ\left(\ide\otimes\tria + \tria\otimes\ide + \stra\otimes\triastra\right) \\[3pt]     
& = m_{A}\circ\Bigg(\overline{T}^{0}(\ide)\otimes\widetilde{T}^0\Big(\tria\Big) + \overline{T}^{0}\Big(\tria\Big)\otimes\widetilde{T}^0(\ide)  
+\overline{T}^{0}\Big(\stra\Big)\otimes\widetilde{T}^0\Big(\triastra\Big)\Bigg) \\[3pt]
& = m_{A}\circ\left(1\otimes\widetilde{T}^0\Big(\tria \Big)  - \phi\Big(\triazero\Big) \otimes 1  -\phi\Big(\strazero\Big)\otimes\phi\Big(\triastra\Big)\right) \\[3pt]
& = m_{A}\circ\left(1\otimes\Big(2-\log(z\bar z)  \Big)  - \Big(-\frac{1}{\varepsilon} \Big) \otimes 1  -1\otimes  \left(\frac{1}{\varepsilon}+2  - \log(z\bar z)\right)\right)+\mathcal{O}(\eps)\\
& =0
\end{align*}
Note that we used $\overline{T}^{0}(\ide)=1$ in the third line.
Performing the product and rearranging in the second last line we also find again the familiar expression of eq. \eqref{eq:trionelooexp} for the leading term in the momentum expansion of the one-loop triangle,
\begin{equation*}
\widetilde{T}^0\left(\tria\right) =\phi\left(\triazero\right)   + \phi\left(\strazero\right)\phi\left(\triastra\right)\,.\\[3pt]
\end{equation*}
We continue with a two-loop example:
\bnon
(\overline{T}^0\star\widetilde{T}^0)\;\left(\dtria\right)=m_{A}\circ(\overline{T}^{0}\otimes \widetilde{T}^0)\circ\; \Delta \left(\dtria\right)
\enon
The coproduct was given in eq.\ \eqref{coproduct:dtriangle}. We thus obtain
\begin{align*}
(\overline{T}^{0}\otimes \widetilde{T}^0)\circ\; \Delta(\ga) &= \overline{T}^{0}(\ide)\otimes\widetilde{T}^0\Big(\dtria\Big) 
+ \overline{T}^{0}\Big(\dtria\Big)\otimes\widetilde{T}^0(\ide) \; \\
&+ \; \overline{T}^{0}\Big(\dstra\Big)\otimes\;\widetilde{T}^0\Big(\codstra \Big)
+(T^{0})\Big(\dtstra\Big)\otimes\;\widetilde{T}^0\Big(\codtstra\Big)\\
&= 1\otimes\Big(6-3\log(z\bar z) + \frac{1}{2}\log^2(z\bar z) \Big) -\Big( \frac{1}{2}-\frac{1}{2\eps}+\frac{1}{2\eps^2}\Big)\otimes 1\\
&-1\otimes \Big(\frac{19}{2}-5\log(z\bar z) + \log^2(z\bar z) 
+\frac{1}{\eps}\big(\frac{5}{2}-\log(z\bar z)\big) +\frac{1}{2\eps^2} \Big) \; \\
&- \Big(-\frac{1}{\eps}\Big)\otimes\;
\Big( \big(4-2 \log(z\bar z)+\frac{1}{2}\log^2(z\bar z)\big)\eps+2-\log(z\bar z)+\frac{1}{\eps}  \Big)+\mathcal{O}(\eps)
\end{align*}
Finally acting with the product $m_A$ we obtain:
\begin{align*}
m_A\circ(\overline{T}^{0}\otimes \widetilde{T}^0)\circ\; \Delta(\ga) = 0 
\end{align*}

Let us now return to the structure of eq. (\ref{loga}). At a first glance it may suggest that $\overline{T}^0$ and $\widetilde{T}^0$ are inverse operations of each other under the $\star$-product. However, this is not exactly true since for non-primitive graphs, 
\begin{equation*}
 (\widetilde{T}^0\star\overline{T}^{0})\;(\Gamma) = \widetilde{T}^0(\ga) - T^0(\Gamma) \neq 0\,.
\end{equation*}
Therefore $\widetilde{T}^0$ and $\overline{T}^0$ are only inverses of each other under certain conditions. More precisely we prove the following theorem for SLD graphs.
\begin{theo}
\label{thm:leadingantipoderel}
For an SLD graph $\Gamma$ the inverse of $\widetilde{T}^0$ in the group of characters is given by
\bnu\label{antip}
(\widetilde{T}^0)^{-1}(\Gamma) = \widetilde{T}^0\circ S(\Gamma)=\overline{T}^0(\Gamma)\,.
\enu
\end{theo}
The proof is presented in appendix \ref{ap:logantiproof}. In fact, using this theorem we can also find the true inverse of $\widetilde{T}^0$; which in general is given by: (theorem \ref{thm:inverse} in appendix \ref{ap:logantiproof})
\bnon
(\widetilde{T}^0)^{-1}(\Gamma) =
\begin{cases}
\overline{T}^0(\Gamma) & \text{if the hard legs are not all joined at the same vertex,}\\
\phi^{-1}(\Gamma) & \text{if all the hard legs are joined at the same vertex,}\\
\ide_A & \text{if } \Gamma = \ide\; ,
\end{cases}
\enon
where 
\bnon
\phi^{-1}=\phi\circ S
\enon
is the inverse of the map $\phi$ under the convolution product. And, the inverse of $T^0$ for any graph is found to be,
\begin{align*}
(T^0)^{-1}(\Gamma)&=T^0\circ S(\ga) = T^0(-\ga -\sum_{\gai\subsetneq\Gamma}S(\gai)\ga/\gai) \nonumber\\
&= -T^0(\ga) -\sum_{\gai\subsetneq\Gamma}T^0(S(\gai))T^0(\ga/\gai))= -T^0(\ga) =\overline{T}^0(\Gamma)\,.
\nonumber
\end{align*}
Here only the first term survives since $\ga/\gai$
only depends on soft scales and acting with $T^0$ thus makes it scaleless. 

Let us now consider the \emph{remainder} $\mathcal{R}$ of the leading order expansion. A compact expression is presented in the following theorem.
\begin{theo}
The leading order remainder for an SLD graph is given by
\bnon
\mathcal{R}^0(\Gamma)=(\overline{T}^0\star \phi) \;(\Gamma)
\enon
\end{theo}
\begin{proof}
To see that $\mathcal{R}^0(\Gamma)$ is indeed the remainder of the leading order expansion, consider 
\begin{align*}
(\overline{T}^0\star \phi)\; (\Gamma)&=m_{A}\circ(\overline{T}^{0}  \otimes \phi)\circ\; \Delta (\ga) \\
&= m_{A}\circ(\overline{T}^{0}\otimes\phi)\circ
\Big( \ide \otimes \Gamma + \sum_{\gai\subseteq\Gamma}  \gai\otimes \Gamma/\gai    \Big)\\
&= m_{A}\circ \Big( \overline{T}^{0}(\ide) \otimes \phi(\Gamma) + \sum_{\gai\subseteq\Gamma}   \overline{T}^{0}(\gai)\otimes \phi(\Gamma/\gai)    \Big)\\
&=\phi(\ga) - \sum_{\gai\subseteq\ga}T^{0}(\gai)\,\widetilde{T}^0(\ga/\gai)\\
&=(\phi-\widetilde{T}^{0})(\Gamma) = \mathcal{R}^0(\Gamma)
\end{align*}
where we used $\widetilde{T}^0(\ga/\gai)=\phi(\ga/\gai)$ in the third line and eq. (\ref{eq:LSDexp}) to get to the last line.
\end{proof}

It is striking that that the Hopf algebraic structure of the remainder, $\mathcal{R}$, appears to closely resemble that of the $R$ operation in renormalisation. Indeed we can identify the $\overline{T}^0$ as the analogue of the counterterm operation. In the context of renormalisation the counterterm operation is often termed as an antipode, or more precisely, a twisted antipode. Since $\overline{T}^0=T^0\circ S$ we can indeed also make this identification here. 

This connection is not an accident. One may think of the appearance of logarithms in the log-power expansion (especially at leading order in the expansion) as being due to the appearance of divergences at the level of the integrand. These divergences are solely responsible for the fact that integration and Taylor expansion do not commute. The expression for the remainder may thus equivalently be derived from an $R$-operation -- or more precisely the $R^*$-operation -- perspective by including the set of counterterms to subtract all possible IR divergences which would emerge in the diagram in the $P\to0$ limit. This approach was indeed taken by Smirnov in refs.~\cite{Smirnov:1990rz, Smirnov:2002pj} to formulate a general proof of the expansion by subgraph.


%


\subsection{A formalism for higher degree and higher expansion orders}
\label{higherorder}

The Hopf algebra formulation  is more complicated when going beyond the logarithmic case. We therefore first review how these subtleties are dealt with in the Connes-Kreimer construction, and also provide an alternative integrand Hopf monoid-formulation.

\paragraph{Beyond the logarithmic case in BPHZ.}


Here we propose an alternative method for dealing with the non-logarithmic counterterms in the $R$-operation that does not require the introduction of a projector pairing. Instead of mapping the graphs directly to their integrated amplitudes, we first map the graphs to their corresponding integrands based on the Feynman rules. One may think that this introduces an ambiguity since there are different representations corresponding to different loop momentum routings. However in the proposed formalism the momentum routing does not actually need to be fixed, and one can instead keep a momentum conserving delta function at every vertex. 

In contrast to the original Hopf algebra formulation it is  important to keep the labelling of the original edges intact in all subgraphs in the coproduct. That is, we need to drop the notion of identifying isomorphic Feynman graphs with each other. The space of graphs is instead those of edge-labelled graphs. This does not play well with the algebra structure, as not all products of all graphs are well defined. Instead only certain products of graphs can be considered. In particular, all those whose product can be associated to what is called a decomposition of a set $I$, for us the set of edges, into two mutually disjoint subsets $S_1$ and $S_2$, such that their disjoint union fulfills $I=S_1 \sqcup S_2$. Such a structure is then no longer a Hopf algebra, but instead is endowed with the structure of a \emph{Hopf monoid on vector species}. For us the species in question is thus that of labelled graphs. Such a structure contains a similar copoduct and antipode as the original Hopf algebra. A brief introduction to Hopf monoids in vector species is provided in Appendix \ref{hopfmon}. 
A more detailed investigation is beyond the present work. However, we have checked that both the coassociativity and antipode proofs in section \ref{sec:hopfproof} and checks provided in section \ref{ex} all go through identically while keeping the edge labelling intact.   

Let us now come back to the integrand construction for the $R$ operation, where we will now use the notation $\mathbf{H}^{\text{1PI}}$ for the relevant Hopf monoid as it is based on 1PI graphs.  For a particular indexing set $I=\{1,...,n\}$ we construct $\mathbf{H}^{\text{1PI}}[I]$ by taking all graphs present in the Connes-Kreimer Hopf algebra of bridgeless graphs, $\mathcal{H}^{\text{1PI}}$, which have precisely $n$ edges, and labelling them in all possible ways.  We then introduce the map $F:\mathbf{H}^{\text{1PI}} \to \mathcal{A}$, with $\mathcal{A}$ the algebra of functions of the integrand and let $\mathcal{I}:\mathcal{A}\to A$ be the corresponding integration map, that integrates the integrand w.r.t.\ the loop momenta present in that integrand. Composing these two maps then yields the previously defined $\phi$-map $\phi=\I\circ F$. 

We can further specify the integrand and integral algebras by including graph-labels to specify the kinematic space. Then, for instance, let $A_\gamma$ be the algebra of functions whose domain is the external kinematic data of the graph $\gamma$, while $\A_\gamma$ is the algebra of functions depending on the internal and external kinematic data of $\Gamma$.  In this language we can then write $Z(\gamma) \in A_{\gamma}$ or equivalently $Z(\gamma)\in \A_{\Gamma/\gamma}$, i.e. the counterterm of a subgraph is a polynomial of the external and internal kinematic data of the contracted parent graph $\Gamma/\gamma$. There is then a product on this space, $m_\A:\A\times \A \to \A$, with which we can write the $R$-operation as:
\bnu \label{renoconvo}
R = \mathcal{I}\circ(Z\star F) = \int d\mu\; (m_\A\circ(Z\otimes F)\circ\Delta_{\mathbf{H}^{\text{1PI}}})\,,
\enu
where $d\mu$ is the integration measure corresponding to the loop momenta of each graph on the right side of each term of the tensor product. Thus, 
\bnu \label{reno}
R(\ga) = \sum_{\gamma\subseteq\ga, \gamma\in\mathbf{H}^{\text{1PI}}} \int d\mu_{\ga/\gamma}\;Z(\gamma)\,F(\ga/\gamma)\,.
\enu
This formulation has several advantages. It allows us to work both with the Hopf-theoretic structure of graphs, while also giving us access to operate at the integrand level. Furthermore it circumvents the need for the projector pairing formalism in the case where the counterterm, $Z(\gamma)$, has non-trivial dependence on the external kinematic data of $\gamma$. 



\paragraph{Beyond the logarithmic case in the expansion by subgraph.}

In the case of asymptotic expansions, we will refer to the corresponding Hopf monoid as the asymptotic Hopf monoid $\mathbf{H}_Q$. For a particular indexing set $I=\{1,...,n\}$ we construct $\mathbf{H}_Q[I]$ by taking all graphs present in the asymptotic Hopf algebra, $\mathcal{H}_Q$, which have precisely $n$ edges, and labelling them in all possible ways.  
We note that there indeed exists a map, \emph{the Fock functor} $\bar{\mathcal{K}}$, see refs.\ \cite{aguiar2017hopfmonoidsgeneralizedpermutahedra,monoidbook},  which allows one to map this Hopf monoid $\mathbf{H}_Q$ back into the Hopf algebra $\mathcal{H}_Q$.

We now express the expansion by subgraph in the form of equation \eqref{reno}. In our discussion here, we will distinguish between Feynman graphs ($\in \mathbf{H}_Q$), Feynman integrands ($\in\mathcal{A}$) and Feynman integrals ($\in A$) based on the notation introduced in section \ref{bg}. It is transparent that the remainder (equation \eqref{remo}) can be written as 
\bnu \label{almostthere}
    \mathcal{R}^{(a)}(\ga) = \int d^Dl\;F(\ga)\:+\sum_{\substack{\gai\subseteq\ga\\ \gai\in\mathbf{H}_Q}} \int d^Dl\;(-\T_F^{\ba})(\gai)\,F(\ga/\gai)\,,
\enu
where $\T_F^{a}:\mathbf{H}_Q\to\A$ is such that $$\T_F^{a}:= \T^{a+\omega}\circ F\,,$$ with $F:\mathbf{H}_Q \to \mathcal{A}$ the function that maps each graph to its corresponding integrand based on the Feynman rules as defined in section \ref{bg}. The integration is over all the loop momenta of $\ga$. 

We now construct $\overline{\T}_F^{a}$ from $\T_F^{a}$ analagously to how $\overline{T}^{a}$ was constructed from $T$ in eq.\ (\ref{eq:Tbar}).
We can then write eq. \eqref{almostthere} as
\bnon
    \mathcal{R}^{(a)}(\ga) = \sum_{\substack{\gai\subseteq\ga, \gai=\ide \\ \gai\in\mathbf{H}_Q}} \int d^Dl\;\overline{\T}_F^{\ba}(\gai)\,F(\ga/\gai)\,.
\enon
The graph combinatorics of the integrand above is now exactly the same as the coproduct of $\mathbf{H}_Q$. Denoting $m_{\A}$ as the multiplication operation for the algebra $\A$, we can write the integrand above as a Hopf convolution product: 
\bnon
\sum_{\substack{\gai\subseteq\ga, \gai=\ide \\ \gai\in\mathbf{H}_Q}} \overline{\T}_F^{\bar{a}}(\gai)\,F(\ga/\gai) \; =\; m_{\A}\circ(\overline{\T}_F^{\ba}\otimes F)\circ\; \Delta (\ga) \;=\;(\overline{\T}_F^{\, \ba}\star F)\;(\ga)
\enon
Now we can write the formulation of the remainder function in Hopf-theoretic language: 
\bnu\label{rem}
\mathcal{R}^a(\ga)\,=\,\I\circ(\overline{\T}_F^{\ba}\star F)\;(\ga)
\enu
This relation holds not only for $\ga$, but also for any AI subgraph of $\ga$ using similar arguments as shown above. In order to get the  expression \eqref{ebrnew}, we act both sides of the equation with the regular Taylor expansion operator $\TT^a$:
\bqa
0 = \TT^a\circ\phi(\gamma) - \TT^a\circ\left(\sum_{\mu_{AI}\subseteq\gamma}T^{\ba+\omega}(\mu_{AI})\ast\,\phi(\gamma/\mu_{AI})\right) \\
\implies \widetilde{T}^a(\gamma) - \sum_{\mu_{AI}\subseteq\gamma}T^{\ba+\omega}(\mu_{AI})\ast\,\widetilde{T}^a(\gamma/\mu_{AI}) = 0
\eqa
which is the expansion by subgraph expression. Notice that equation \eqref{rem} does not hold true for graphs in $\mathbf{H}_Q$ that do not depend on the hard scales. These graphs are those that have all the hard legs joined at the same vertex. This is natural because there is no meaning of expansion by subgraph if there is no hierarchy of scales.

\paragraph{The twisted antipode}
It is instructive to study the whole series, i.e. the result of the operator $\widetilde{T}^\infty$. Within the radius of convergence \cite{Smirnov:1990rz, Smirnov:1994tg, Smirnov:2002pj} the resulting series therefore reproduces the full function, and the remainder of the series vanishes. In this case eq. (\ref{rem}) leads to
\begin{equation}
   0=\,\I\circ(\overline{\T}^{\infty}_F\star F)\;(\ga)\,.
\end{equation}
For a non-trivial graph it thus appears that $\overline{\T}^{\infty}_F$ acts like the inverse of the Feynman-rule map $F$ in the group of characters, at least when acting on MHM graphs, and with the product 
being integrated with $\I$. In this setting $\overline{\T}^{\infty}_F$ thus has the same action as $F\circ S$, with $S$ the antipode. This is true within the convergence domain of the expansion parameters, but as for some AI subgraphs, loop momenta are also soft parameters, integration takes them outside the convegence domain, so we cannot generally identify $\overline{\T}^{\infty}_F$ with the true antipode.  Nevertheless we can still relate $\overline{\T}^{\infty}_F$ to a twisted antipode, in analogy to the counterterm map in the Connes-Kreimer construction. Using Birkhoff decomposition, see appendix \ref{Birkh}, we arrive at the following expression  for the \emph{twisted integrand antipode}
\begin{equation}
S_F(\gamma) := -F(\gamma) - \sum_{\mu\subsetneq\gamma}\T^{\infty}\circ S_F(\mu)\;F(\gamma/\mu)\,.
\end{equation}
The corresponding \emph{integrated} twisted antipode, or just the twisted antipode, is then defined as
$$S_\phi:=\I\circ S_F\,.$$
With this definition we then arrive, see appendix \ref{Birkh} for the derivation, at the following proposition:
\begin{theoremEnd}{theo}
\label{thm:twistantipode}
For any hard mass-momentum spanning graph in $\mathbf{H}_Q$, under dimensional regularisation, the twisted antipode is the negative of the standard Taylor expansion (the Taylor expansion in Smirnov's original definition \ref{ebrnew}):
\bnu\label{twistantipode}
S_\phi = -T^{\infty}
\enu
\end{theoremEnd} 
\begin{proofEnd}
We begin by noticing that the Taylor expansion to all orders, $\T^\infty:\A\to\A$ is a Rota-Baxter operator, as within the convergence domain of the soft parameters it is equal to $id_{\A}$ which is trivially a Rota-Baxter operator. However, $id_{\A}\neq \T^\infty$ as some of the soft parameters are loop momenta which under integration take up values beyond their convergence domain.

We will now use theorem \ref{thm:BHd} to Birkhoff decompose the Feynman integrand map $F$ using the projection operator $\T^\infty$. This yields
\begin{equation}
\label{eq:fminus}
\begin{split}
F_-(\gamma) &= - \T^\infty\big(F(\gamma) + \sum_{\mu\subsetneq\gamma}F_-(\mu)F(\gamma/\mu)\big) \,, \\
F_+(\gamma) &= (id_A- \T^\infty)\big(F(\gamma) + \sum_{\mu\subsetneq\gamma}F_-(\mu)F(\gamma/\mu)\big).  \end{split}
\end{equation}
Let us focus at the term in the brackets for both the equations. This term
\bnu
\label{eq:Fsum}
F(\gamma) - \sum_{\mu\subsetneq\gamma}F_-(\mu)\;F(\gamma/\mu)
\enu
is very similar, for what concerns the graph combinatorics, to the antipode acted upon by $F$:
\bnu
\label{eq:FSsum}
F\circ S\,(\gamma) = -F(\gamma) - \sum_{\mu\subsetneq\gamma}F \circ S(\mu)\;F(\gamma/\mu)\,.
\enu
But \eqref{eq:Fsum} and \eqref{eq:FSsum} are not equal to each other. $F_-$ is acted over by $\T^\infty$, and  outside the convergence domain does not equal to $F\circ S$, which becomes relevant during integration over loop momenta. Thus, taking inspiration from the terminology in \cite{Connes:1999yr}, we define the twisted \emph{integrand} antipode $S_F$:
\bnu\label{eq:twisted}
S_F(\gamma) := -F(\gamma) - \sum_{\mu\subsetneq\gamma}\T^\infty\circ S_F(\mu)\;F(\gamma/\mu)\,,
\enu
With the twisted integrand antipode we can write $F_-$ in a compact form: 
\bnon
F_- = \T^\infty\circ S_F.
\enon
The above statement can be proven using induction on the grade of the graph with the same method shown in appendix \ref{ap:logantiproof}. \\
Fom the Birkhoff decomposition of $F$ we have
\bnon
F=(F_-)^{-1}\star F_+ \; \implies\;F_+=F_-\star F.
\enon 
Integrating both sides of the equation, we get, for hard mass-momentum spanning graphs $\gamma\in\mathbf{H}_Q$,
\bqa
\I\circ F_+\,(\gamma)&=&\I\circ(F_-\star F)\,(\gamma)\\
\implies-\I\circ S_F\,(\gamma)+\I\circ\T^\infty\circ S_F\,(\gamma)&=&\I\circ(\overline{\T}^{\infty}_F\star F)\,(\gamma)
\eqa
where $\overline{\T}^{\infty}_F$ was defined in section \ref{higherorder}. The RHS of the second line is so because the $\T^\infty$ of $F_-$ converts all contracted graphs in $S_F$ to scaleless graphs, which go to zero after integration. Thus only the first term in $F_-(\mu)$ survives, and we have only $\overline{\T}^{\infty}_F(\mu)$ that really goes into the RHS after integration. Following the same reasoning, we have $\I\circ\T^\infty\circ S_F\,(\gamma) = -\I\circ\T^{\infty}\,(\gamma)$. Using $T^{\infty}=\I\circ\T_F^{\infty}$, we have, 
\bnon
-S_\phi(\gamma) - T^{\infty}(\gamma) = \mathcal{R}^\infty(\gamma)=0
\enon
where $S_\phi:=\I\circ S_F$ is the twisted \emph{integral} antipode or just the twisted antipode. This proves Theorem \ref{thm:twistantipode}.
\end{proofEnd}
The theorem relates a recursive operator (LHS) to a straightforward Taylor expansion (RHS). This explains the apparent lack of a recursion in the expansion by subgraph even though it has a Hopf algebraic antipode structure. 
It would be interesting to consider the expansion-by-subgraph in non-analytic regularisation schemes where scaleless integrals do not vanish. In such schemes the expansion by subgraph would then have to incroporate the more involved recursive structure of the twisted antipode, which is absent in dimensional regularisation.

\section{Discussion and Conclusions}
\label{conclusion}
In this paper we have established a new Hopf algebra framework for an entire class of asymptotic Feynman integral expansions, namely all small/large mass and momentum expansions which can be defined in Euclidean regime. Since the early nineties due to works by Smirnov this class of expansions have been understood also in an \emph{expansion by subgraph} approach, which identifies regions with certain subgraphs. In contrast, for other expansions, which can only be defined in the Minkowkian regime, the subgraph approach has only recently been fully established for certain special cases, and is more complicated \cite{Gardi:2022khw, Herzog:2023sgb, Ma:2023hrt}. We leave a possible Hopf-algebraic description of these more general cases to future work. 

For what concerns the Euclidean expansions, the subgraphs are ultimately connected to Euclidean UV and IR divergences. The combinatorics of these divergences have, over the last decade, been understood in terms of the motic Hopf algebra by Brown ~\cite{Brown:2015fyf}, which also served in the reformulation of the $R^*$ operation ~\cite{Beekveldt:2020kzk}, which subtracts the corresponding divergences. The first step of this paper was the establishment of an \emph{asymptotic Hopf algebra}, closely related to the motic Hopf algebra, but whose graph-combinatorics depends, critically, on the identification of a set of small/expansion parameters. We established a natural coproduct, which sums over the set of subgraphs generated in a particular such expansion, and showed that it satisfies the requirements of a Hopf algebra in its own right. As in the case of the Connes-Kreimer and motic Hopf algebras the antipode follows from the fact that the Hopf algebra is graded and connected to the identity. 

We proceeded to establish a Hopf algebraic formulation of the expansion by subgraph. Before attacking the all-order case we focused on the leading-power term in the expansion in the special, although not uncommon, case that the expansion degrees of all contributing subgraphs are logarithmic. In this case the Hopf algebraic formulation is the cleanest. In Theorem \ref{thm:logAHopfalgebra} we show that the expansion by subgraph follows from a simple convolution product of the region-expansion (or regular) Taylor operator and standard Taylor expansion operators in the group of characters of the Hopf algebra. We illustrate how our formulation works with a detailed example.  

We further establish a deep relation between the region expansion operator and the standard Taylor expansion operator through the antipode in Theorem  \ref{thm:leadingantipoderel}. A particularly insightful result is the formulation of the remainder of the leading term of the expansion. We provide it in terms of a convolution of the standard Taylor expansion operator and the Feynman rules, expressed as a map. What is striking about this equation is its similarity with the Connes-Kreimer Hopf-algebraic formulation of renormalisation, which takes the same form if the standard Taylor operator is replaced with the counterterm, or twisted antipode, operation. 

We furthermore present an extension of the framework to higher power terms in the expansion, which we provide in terms of convolutions of maps at the integrand level. This formulation requires the labels of edges in different subgraphs to remain intact, a feature which is usually lost in the Hopf algebra formulation. Instead we introduce a new Hopf monoid formulation, which allows keeps the edge labels. This provides an alternative to the projector pairing formalism, which is usually employed to factorise elements in the convolution product. Within the Hopf monoid framework we derive an interesting result showing that the standard Taylor operator can be identified with a twisted antipode, which we motivate from a Birkhoff decomposition of the integrand. 

There remains much to be done in the Hopf-algebraic formulation of the expansion by region. For the Euclidean case we only took a first step here. It would be desirable to give a complete proof of the expansion by subgraph within the Hopf-theoretic framework, this goes well beyond what we have achieved here - but some of the results we presented, such as the relationship of the Taylor operator with the twisted antipode, may be of use. Indeed the fact that the remainder of the expansion takes a similar form as the $R$ operation, which was also a key element in Smirnov's original proof, should allow for the entire Hopf algebraic machinery developed by Kreimer, Connes and others to be employed here. 

Another promising direction will be to explore the Hopf monoid sructure. It should allow to make contact also with the Hopf monoid structure of the generalised permutahedron \cite{aguiar2017hopfmonoidsgeneralizedpermutahedra}, which is known to exist on the Feynman polytope in the Euclidean regime \cite{Schultka:2018nrs, Borinsky:2020rqs, Panzer:2019yxl}. Finally this may allow to shed light into the mathematical structure for Minkowskian expansions, where the corresponding Feynman polytope is more involved than a generalised permutahedron.

\acknowledgments
We would like to thank Michael Borinsky for early collaboration on this work. MC would like to thank Balasubramanian Ananthanarayan for continued support and encouragement and Sudeepan Datta for insightful discussions. This work is supported by the UKRI FLF grant ``Forest Formulas for the LHC'' (Mr/S03479x/1), and the ERC Advanced Grant 101095857 Conformal-EIC.

\newpage

\appendix
\section{Antipodal relations for the SLD graphs}\label{ap:logantiproof}
We present the proof for theorem \ref{thm:leadingantipoderel} in this appendix. 
\begin{proof}
As discussed below equation \eqref{eq:convproduct}, the inverse of any character $g$ in the group of characters is given by $g\circ S$. Thus, for any $\gamma\in\hopf$, we have $$(\widetilde{T}^0)^{-1}(\gamma) =\widetilde{T}^0\circ S(\gamma).$$ Now, the antipode is given by both the left and right antipodes: $$S(\gamma)\; =\; -\gamma -\sum_{\mu\underset{AI}{\subsetneq}\gamma} S(\mu)\,\gamma/\mu \;=\; -\gamma -\sum_{\mu\underset{AI}{\subsetneq}\gamma} \mu\,S(\gamma/\mu)\, .$$ The properties of Hopf algebra, wherein both the left and right antipodes are equal, makes it sufficient for us to prove the theorem for only one case. The case of the left antipode is more convenient: $$\widetilde{T}^0\circ S(\gamma)\; =\; -\widetilde{T}^0(\gamma) -\sum_{\mu\underset{AI}{\subsetneq}\gamma} \widetilde{T}^0\circ S(\mu)\,\widetilde{T}^0(\gamma/\mu)\, .$$ We can use the principle of mathematical induction on the grading of the Hopf algebra to prove that $\widetilde{T}^0\circ S(\gamma) = \overline{T}^0(\gamma)$ for any SLD graph $\gamma$. This is achieved as follows: with the reduced coproduct given by $$\widetilde{\Delta}(\gamma)=\Delta(\gamma)-\ide\otimes\gamma-\gamma\otimes\ide\, ,$$ define the \emph{grading operator}:$$\mathcal{P}_n:=\widetilde{\Delta}^{\otimes n}=\widetilde{\Delta}\otimes\ldots\otimes\widetilde{\Delta} \; \text{($n$ times)}\;.$$ Then, the \emph{grade} of a graph $\gamma\in\hopf$ is the minimum value of $n$ for which $\mathcal{P}_n(\gamma)=0$. Notice that, all proper AI subgraphs of a given $\gamma\in\hopf$ have grade strictly less than the grade of $\gamma$. This helps us to set up an inductive proof.

Let us take our inductive hypothesis to be that $\widetilde{T}^0\circ S(\gamma) = \overline{T}^0(\gamma)$ for any SLD graph $\gamma$. For grade $n=1$, i.e.\ primitive graphs, $S(\gamma)=-\gamma$, and $\widetilde{T}^0(\gamma)=T^0(\gamma)$, thus $$\widetilde{T}^0\circ S(\gamma)=-\widetilde{T}^0(\gamma)=-T^0(\gamma)=\overline{T}^0(\gamma).$$ Now let us assume that, given an arbitrary positive integer $N$, for any SLD graph with grade $n\leq N$ the inductive hypothesis holds. Let $\gamma$ be an SLD graph with grade $N+1$. Then,
\begin{align*}
    \widetilde{T}^0\circ S(\gamma) & =\, -\widetilde{T}^0(\gamma) -\sum_{\mu\underset{AI}{\subsetneq}\gamma} \widetilde{T}^0\circ S(\mu)\,\widetilde{T}^0(\gamma/\mu)\\
    & = -\widetilde{T}^0(\gamma) -\sum_{\mu\underset{AI}{\subsetneq}\gamma} \overline{T}^0(\mu)\,\phi(\gamma/\mu)\\
    & = -\widetilde{T}^0(\gamma) +\sum_{\mu\underset{AI}{\subsetneq}\gamma} T^0(\mu)\,\phi(\gamma/\mu)\\
    & = \overline{T}^0(\gamma)
\end{align*}
where the second step follows from the fact that the $\mu$ are all proper AI subgraphs of $\gamma$, which are also SLD as $\gamma$ is SLD. Thus, they all have grade $n\leq N$, and the inductive hypothesis can be applied to them. By definition of SLD graphs, $\gamma/\mu$ is of logarithmic SDD, and as it depends purely on soft scales, $\widetilde{T}^0(\gamma/\mu) = \phi(\gamma/\mu)$. The last step then follows from expansion by subgraph, confirming the inductive hypothesis.

This by itself is sufficient to prove theorem \ref{thm:leadingantipoderel}. The combinatorics of the Hopf algebraic structure ensures that the right antipode relation holds. 
\end{proof}
\begin{theo}\label{thm:inverse}
    The inverse of $\widetilde{T}^0$ for any general $\Gamma\in\hopf$ in the logarithmic case is given by 
\bnon
(\widetilde{T}^0)^{-1}(\gamma) =
\begin{cases}
\overline{T}^0(\Gamma) & \text{if $\Gamma$ inherits all hard scales (the SLD graphs)}\\
\phi^{-1}(\Gamma) & \text{if $\Gamma$ inherits only soft scales}\\
\ide_A & \text{if } \Gamma = \ide
\end{cases}
\enon
where 
\bnon
\phi^{-1}(\Gamma)=\phi\circ S(\Gamma).
\enon
\end{theo}
\begin{proof}
We have already proven the first case, i.e.\ the case for SLD graphs, and the case for the empty graph is trivial.

Now, let us consider the case of graphs that only inherit the soft scales. Graph theoretically speaking, they refer to the elements in $\hopf$ obtained by contracting AI subgraphs. These graphs form a sub-Hopf-algebra as the product, coproduct and antipode operations on such graphs are always closed in graphs that only depend on soft scales. And this sub-Hopf-algebra inherits the convolution product from $\hopf$. As inside this sub-Hopf-algebra $\widetilde{T}^0 = \phi$ identically, the inverse of $\widetilde{T}^0$ is given by the inverse of $\phi$, which is $\phi^{-1}(\Gamma)=\phi\circ S(\Gamma)$. This completes the proof.
\end{proof}

\section{Twisted Antipode and Birkhoff decomposition}
\label{Birkh}

The role of the antipode at higher orders in the expansion by subgraph can be studied using Birkhoff decomposition, which decomposes any homomorphism acting on the Hopf algebra into a convolution product of two homomorphisms based on a projection operator. A similar treatment was used in the Connes-Kreimer Hopf algebra \cite{Connes:1999yr} to establish the counterterm operation as a twisted antipode.
\begin{theo}[Birkhoff decomposition]
\label{thm:BHd}
Let $g:\hopf\to\A$ be a character and let $K$ be an endomorphism in $\A$ such that the following relation
\bnu\label{tayl}K(a)K(b)=K(K(a)b+aK(b)-ab)\enu
holds for any two elements $a,b\in\A$ and any positive integers $s,t$. Then, there exists a unique Birkhoff decomposition of $g$:
$$g=g_-^{-1}\star g_+$$
where both $g_-,g_+\in \mathcal{G}$ and $g_-^{-1}$ is the inverse of $g_-$ under the convolution product. Furthermore, $g_-$ and $g_+$ are given by:
\bqa
g_-&=&-K(g_-\star g\circ(id-e))\\
g_+&=&(id_{\A}-K)(g_-\star g\circ(id-e))
\eqa
where $id_{\A}$ is the identity operator for the algebra $\A$ and maps any element of $\A$ to itself. $id$ is the familiar identity operator for $\hopf$, and $e=u\circ\bar{e}$. The solutions to the Birkhoff decomposition are unique.
\end{theo}
We are going to assume, in the following, that theorem \ref{thm:BHd} also applies to the Hopf monoid $\mathbf{H}_Q$. The proof is likely a straightforward extension which, however, we do not attempt here. 

\printProofs

\noindent The twisted antipode can also be defined at any finite order of expansion, using a modified form of Birkhoff decomposition with Taylor operators laid forth by Manchon and Mohamed \cite{manchon2014bialgebra}, which modifies the Rota-Baxter relation to $$\P_s(a)\P_t(b)=\P_{s+t}(\P_s(a)b+a\P_t(b)-ab)$$ where $\P$ is an indexed family of endomorphisms in $\A$. Following very similar steps as above, we can obtain the twisted integrand antipode up to order $\ba$ in the inverse powers of hard scales:
$$S_F^{\ba}(\gamma) := -F(\gamma) - \sum_{\mu\subsetneq\gamma}\P_{\ba}\circ S_F^{\ba}(\mu)\;F(\gamma/\mu)$$
where $\P_{\ba}$ is the operator $\T^{\ba+\omega}$. We skip the details here as this is not relevant to any of the theorems we introduce in this paper. From here the twisted antipode up to the order $\ba$ in inverse hard scales is given as $$S_\phi^{\ba}(\gamma)=\I\circ S_F^{\ba}(\gamma)=-\mathcal{R}^{\ba+\omega}(\gamma)-T^{\ba+\omega}(\gamma)$$ for all AI subgraphs $\gamma\subseteq\ga$, where, as before, $\omega$ is the UV SDD of $\gamma$. This order by order twisted antipode might be useful for a Hopf theoretic proof of expansion by subgraph in the future.

\section{Hopf Monoids: A Brief Introduction}\label{hopfmon}
As was discussed in section \ref{higherorder}, mapping the graphs to the integrands instead of the integrals requires us to introduce a labelling for the internal edges in order to assign the respective momenta. We then use the Dirac delta functions at the vertices to impose conservation of momentum. Finally we integrate over the momenta of each internal edge to get the Feynman integral. However, such a labelling of internal edges leads to ill-defined products of graphs. Take for example, the parent graph
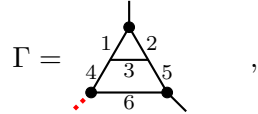
\bnon
\ga = \dtriain\qquad, 
\enon
and consider the following two subgraphs
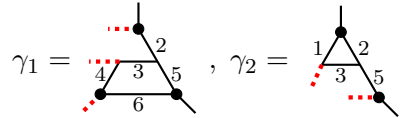
\bnon
\gamma_1=\dtriaaa\;, \;\gamma_2 = \dtstrain
\enon
Both $\gamma_1,\gamma_2\in \mathcal{H}_\ga$ as they are AI subgraphs of $\ga$. We had defined multiplication in the Hopf algebra as the disjoint union of graphs, a notion that only holds for unlabelled graphs. Thus, we have to come up with a new defintion of a product, which also has to be consistent with the integrand map $F$. However, multiplying the integrands of maps like $\gamma_1$ and $\gamma_2$ which share common labellings poses a problem: because of the multiple appearance of the edges, we end up with extra Dirac delta functions that cannot be integrated out using the integration operator $\I$, which is a crucial step in mapping the integrands to the integrals.

This problem for the product, though, turns out to be irrelevant for asymptotic expansions and the local $R$-operation. Indeed, the only relevant quantities there correspond to the coproduct and antipode acted on individual graphs, and in those operations, multiplications of graphs sharing common labellings never occurs. This is because the multiplications in the coproduct and antipode are always between subgraphs and their corresponding contracted graphs, and they always have a mutually disjoint set of edge labels. 

Therefore, having a full algebra structure is not necessary for a Hopf-theoretic formulation of asymptotic expansions, and there indeed exists a mathematical structure that helps us retain the required coproduct and antipode while avoiding problematic products as shown above. That structure is a Hopf monoid on vector species.

Here we will only be giving a very elementary introduction to Hopf monoids, in order to establish their relevance in dealing with graphs carrying internal edge-labellings. Following \cite{aguiar2017hopfmonoidsgeneralizedpermutahedra}, we will begin our discussion with set species. We can then linearise the sets to construct the relevant vector spaces.\\

\noindent A \emph{set species} $\rh$ consists of the following data:
\begin{itemize}
    \item For every finite set $I$, a corresponding set $\mathrm{H}[I]$,
    \item For any bijection $\sigma:I\to J$, there exists a map $\rh[\sigma]:\rh[I]\to\rh[J]$ such that $\rh[\sigma\circ\tau]=\rh[\sigma]\circ\rh[\tau]$ and $\rh[id]=id$.
\end{itemize}
In our applications, the set $I=\{1,\dots,n\}$ is the set of edge labellings and $\rh[I]$ contains the set of all possible graphs with $n$ edges of a certain kind  (for us those which were previously present in the Hopf algebra under consideration) with all possible edge labellings $I$. 


Next we define the decomposition of $I$ as the finite \emph{sequence} $(S_1,\ldots, S_k)$ of mutually exclusive sets $S_i$ such that $$I=S_1\sqcup\ldots\sqcup S_k\,.$$ Note that under this definition $I=S\sqcup T$ and $I=T\sqcup S$ are distinct decompositions of $I$ as they correspond to different sequences $(S,T)$ and $(T,S)$, unless $S=T=I=\emptyset$.\\[2pt]
Based on this, we can define a connected Hopf monoid on set species as containing the following structures:
\begin{itemize}
    \item A set species $\rh$ such that $\rh[\emptyset]$ is a singleton,
    \item For each set $I$ and each decomposition $I=S\sqcup T$, product and coproduct maps
\begin{equation}
\label{eq:hopfmonoidoperations}
    \rh[S]\times\rh[T]\xrightarrow{\mu_{S,T}}\rh[I] \qquad \rh[I]\xrightarrow{\Delta_{S,T}}\rh[S]\times\rh[T]\,,
\end{equation}
    where the product is defined as $\mu_{S,T}(x,y) = x|y$ , with $x\in S$ and $y\in T$,\\
    and the coproduct $\Delta_{S,T}(z)=(z_{|S}, z_{/S})$, with $z\in \rh[I]$, $z_{|S}\in S$ the restriction of $z$ in S, and $ z_{/S}\in T$ the contraction of $S$ from $z$.
\end{itemize}
With additional conditions of associativity, unitality etc on the product and coproduct, (the details of which we omit here) this becomes a Hopf monoid on the set species $\rh$.

Notice that the decomposition of the label set $I$ into mutually exclusive sets forms the basis of the definition, and thus completely avoids the problematic cases from before, as now we only consider products of elements having no common labellings. However, the set species does not contain an operation for addition which we need in our coproduct, and neither does it come with an antipode. For that we need to linearise the Hopf monoid on the set species $H$ into a Hopf monoid on the vector species $\mathbf{H}$, by using the $H[I]$'s as basis sets of vector spaces $\mathbf{H}[I]$, and correspondingly linearising the product and coproduct. The Cartesian product `$\times$' in \eqref{eq:hopfmonoidoperations} is then promoted to a tensor product `$\otimes$'. We can also define an antipode which leads us to the same group structure under convolution that is key to using tools like Birkhoff decomposition.\\

 \bibliographystyle{JHEP}
 \bibliography{biblio.bib}






\end{document}

%% file: pics.tex
\usepackage[utf8]{inputenc}
\usepackage{xcolor}
\usepackage{xfrac}
\usepackage{amsmath}
\usepackage{amsfonts}
\usepackage{amsthm}
\usepackage{amssymb}
\usepackage{accents}
\usepackage{feynmf}
\usepackage{tikz}
\usetikzlibrary{shapes.geometric, arrows, patterns}
\usetikzlibrary{calc}
\tikzstyle{top} = [rectangle, rounded corners, minimum width=3cm, minimum height=1cm, text centered, text width=3cm, draw=black, fill=red!30]
\usepackage{pgfplots}
\pgfplotsset{compat=1.11}
\usepgfplotslibrary{fillbetween}
\usetikzlibrary{intersections}

\pgfdeclarelayer{bg}
\pgfsetlayers{bg,main}
\usepackage{hyperref}

\newcommand{\triazero}{
\begin{minipage}{1.5cm}
\begin{tikzpicture}
\draw [thick] (0,0.866) -- (-0.5,0);
\draw [thick] (0.5,0) -- (0,0.866);
\draw [thick] (0.5,0) -- (0.75,-0.25);
\draw [thick] (-0.5,0) -- (-0.75,-0.25);
\draw [thick] (-0.5,0) -- (0.5,0);
\filldraw (-0.5,0) circle (2pt);
\filldraw (0.5,0) circle (2pt);
\filldraw (0, 0.866) circle (2pt);
\draw (0,1.216) node{ };
\end{tikzpicture}
\end{minipage}}

\newcommand{\dtriain}{\begin{minipage}{1.6cm}
\begin{tikzpicture}
\draw [thick] (0,0.866) -- (-0.5,0);
\draw [thick] (0,0.866) -- (0,1.216);
\draw [thick] (0.25, 0.433) -- (-0.25, 0.433);
\draw [thick] (0.5,0) -- (0,0.866);
\draw [thick] (0.5,0) -- (0.75,-0.25);
\draw [ultra thick, red, dotted] (-0.5,0) -- (-0.75,-0.25);
\draw [thick] (-0.5,0) -- (0.5,0);
\filldraw (-0.5,0) circle (2pt);
\filldraw (0.5,0) circle (2pt);
\filldraw (0, 0.866) circle (2pt);
\draw (-0.3,0.65) node{${\scriptstyle 1}$};
\draw (0.3,0.65) node{${\scriptstyle 2}$};
\draw (0,0.3) node{${\scriptstyle 3}$};
\draw (0.5,0.27) node{${\scriptstyle 5}$};
\draw (-0.5,0.27) node{${\scriptstyle 4}$};
\draw (0,-0.13) node{${\scriptstyle 6}$};
\end{tikzpicture}
\end{minipage}}

\newcommand{\dtriaaa}{\begin{minipage}{1.6cm}
\begin{tikzpicture}
\draw [thick] (-0.25,0.433) -- (-0.5,0);
\draw [thick] (0,0.866) -- (0,1.216);
\draw [thick] (0.25, 0.433) -- (-0.25, 0.433);
\draw [thick] (0.5,0) -- (0,0.866);
\draw [thick] (0.5,0) -- (0.75,-0.25);
\draw [ultra thick, red, dotted] (-0.5,0) -- (-0.75,-0.25);
\draw [ultra thick, red, dotted] (-0.25, 0.433) -- (-0.65, 0.433);
\draw [ultra thick, red, dotted] (0,0.866) -- (-0.4,0.866);
\draw [thick] (-0.5,0) -- (0.5,0);
\filldraw (-0.5,0) circle (2pt);
\filldraw (0.5,0) circle (2pt);
\filldraw (0, 0.866) circle (2pt);
\draw (0.3,0.65) node{${\scriptstyle 2}$};
\draw (0,0.3) node{${\scriptstyle 3}$};
\draw (0.5,0.27) node{${\scriptstyle 5}$};
\draw (-0.5,0.27) node{${\scriptstyle 4}$};
\draw (0,-0.13) node{${\scriptstyle 6}$};
\end{tikzpicture}
\end{minipage}}

\newcommand{\dtstrain}{\begin{minipage}{1.2cm}
\begin{tikzpicture}
\draw [ thick] (0,0.866) -- (-0.25,0.433);
\draw [thick] (0,0.866) -- (0,1.216);
\draw [thick] (0.25, 0.433) -- (-0.25, 0.433);
\draw [ultra thick, red, dotted] (-0.25,0.433) -- (-0.4, 0.103);
\draw [thick] (0.5,0) -- (0,0.866);
\draw [thick] (0.5,0) -- (0.75,-0.25);
\draw [ultra thick, red, dotted] (0.1,0) -- (0.5,0);
\filldraw [white] (-0.5,0) circle (2pt);
\filldraw (0.5,0) circle (2pt);
\filldraw (0, 0.866) circle (2pt);
\draw (-0.3,0.65) node{${\scriptstyle 1}$};
\draw (0.3,0.65) node{${\scriptstyle 2}$};
\draw (0,0.3) node{${\scriptstyle 3}$};
\draw (0.5,0.27) node{${\scriptstyle 5}$};
\end{tikzpicture}
\end{minipage}}

\newcommand{\dtria}{\begin{minipage}{1.6cm}
\begin{tikzpicture}
\draw [thick] (0,0.866) -- (-0.5,0);
\draw [thick] (0,0.866) -- (0,1.216);
\draw [thick] (0.25, 0.433) -- (-0.25, 0.433);
\draw [thick] (0.5,0) -- (0,0.866);
\draw [thick] (0.5,0) -- (0.75,-0.25);
\draw [ultra thick, red, dotted] (-0.5,0) -- (-0.75,-0.25);
\draw [thick] (-0.5,0) -- (0.5,0);
\filldraw (-0.5,0) circle (2pt);
\filldraw (0.5,0) circle (2pt);
\filldraw (0, 0.866) circle (2pt);
\end{tikzpicture}
\end{minipage}}

\newcommand{\dtrian}{\begin{minipage}{1.6cm}
\begin{tikzpicture}
\draw [thick] (0,0.866) -- (-0.5,0);
\draw [thick] (0,0.866) -- (0,1.216);
\draw [thick] (0.25, 0.433) -- (-0.25, 0.433);
\draw [thick] (0.5,0) -- (0,0.866);
\draw [thick] (0.5,0) -- (0.75,-0.25);
\draw [ultra thick, red, dotted] (-0.5,0) -- (-0.75,-0.25);
\draw [thick] (-0.5,0) -- (0.5,0);
\filldraw (-0.5,0) circle (2pt);
\filldraw (0.5,0) circle (2pt);
\filldraw (0, 0.866) circle (2pt);
\draw (0,1.416) node{${\scriptstyle 3}$};
\draw (0.9,-0.4) node{${\scriptstyle 2}$};
\draw (-0.9,-0.4) node{${\scriptstyle 1}$};
\end{tikzpicture}
\end{minipage}}

\newcommand{\dstra}{\begin{minipage}{1.2cm}
\begin{tikzpicture}
\draw [ultra thick, red, dotted] (0,0.866) -- (-0.4,0.866);
\draw [thick] (0,0.866) -- (0,1.216);
\draw [ultra thick, red, dotted] (0.25, 0.433) -- (-0.15, 0.433);
\draw [thick] (0.5,0) -- (0,0.866);
\draw [thick] (0.5,0) -- (0.75,-0.25);
\draw [ultra thick, red, dotted] (0.1,0) -- (0.5,0);
\filldraw [white] (-0.5,0) circle (2pt);
\filldraw (0.5,0) circle (2pt);
\filldraw (0, 0.866) circle (2pt);
\end{tikzpicture}
\end{minipage}}

\newcommand{\codstra}{\begin{minipage}{1.0cm}
\begin{tikzpicture}
\draw[color=black, thick, -] (0,0.5) arc (90:450:0.5);
\draw [ultra thick, red, dotted] (0, 0.5) -- (0, 0.9);
\draw [thick] (0, -0.5) -- (-0.3, -0.8);
\draw [ thick] (0, -0.5) -- (0.3, -0.8);
\draw [thick] (0,-0.5) .. controls (0.17,-0.17) .. (0.5,0);
\filldraw (0,0.5) circle (2pt);
\filldraw (0,-0.5) circle (2pt);
\end{tikzpicture}
\end{minipage}}

\newcommand{\codstras}{\begin{minipage}{1.3cm}
\begin{tikzpicture}
\draw[color=black, thick, -] (0,-0.5) arc (270:360:0.5);
\draw [ultra thick, red, dotted] (0, -0.5) -- (-0.4, -0.5);
\draw [thick] (0, -0.5) -- (-0.3, -0.8);
\draw [ thick] (0, -0.5) -- (0.3, -0.8);
\draw [thick] (0,-0.5) .. controls (0.17,-0.17) .. (0.5,0);
\draw [red,  thick] (0.5,0) -- (0.5,0.4);
\filldraw (0.5,0) circle (2pt);
\filldraw (0,-0.5) circle (2pt);
\end{tikzpicture}
\end{minipage}}

\newcommand{\dtstra}{\begin{minipage}{1.2cm}
\begin{tikzpicture}
\draw [ thick] (0,0.866) -- (-0.25,0.433);
\draw [thick] (0,0.866) -- (0,1.216);
\draw [thick] (0.25, 0.433) -- (-0.25, 0.433);
\draw [ultra thick, red, dotted] (-0.25,0.433) -- (-0.4, 0.103);
\draw [thick] (0.5,0) -- (0,0.866);
\draw [thick] (0.5,0) -- (0.75,-0.25);
\draw [ultra thick, red, dotted] (0.1,0) -- (0.5,0);
\filldraw [white] (-0.5,0) circle (2pt);
\filldraw (0.5,0) circle (2pt);
\filldraw (0, 0.866) circle (2pt);
\end{tikzpicture}
\end{minipage}}

\newcommand{\codtstra}{\begin{minipage}{1.0cm}
\begin{tikzpicture}
\draw[color=black, thick, -] (0,0.5) arc (90:450:0.5);
\draw [ultra thick, red, dotted] (0, 0.5) -- (0, 0.9);
\draw [thick] (0, -0.5) -- (-0.3, -0.8);
\draw [ thick] (0, -0.5) -- (0.3, -0.8);
\filldraw (0,0.5) circle (2pt);
\filldraw (0,-0.5) circle (2pt);
\end{tikzpicture}
\end{minipage}}

\newcommand{\tria}{\begin{minipage}{1.6cm}
\begin{tikzpicture}
\draw [thick] (0,0.866) -- (-0.5,0);
\draw [ultra thick, red, dotted] (0,0.866) -- (0,1.216);
\draw [thick] (0.5,0) -- (0,0.866);
\draw [thick] (0.5,0) -- (0.75,-0.25);
\draw [thick] (-0.5,0) -- (-0.75,-0.25);
\draw [thick] (-0.5,0) -- (0.5,0);
\filldraw (-0.5,0) circle (2pt);
\filldraw (0.5,0) circle (2pt);
\filldraw (0, 0.866) circle (2pt);
\end{tikzpicture}
\end{minipage}}

\newcommand{\stra}{\begin{minipage}{2.0cm}
\begin{tikzpicture}
\draw (0,0) -- (-1,0);
\draw (0,0) -- (0.35,-0.35);
\draw (-1,0) -- (-1.35,-0.35);
\draw [ultra thick, red, dotted](-1,0) -- (-1.35,0.35);
\draw [ultra thick, red, dotted](0,0) -- (0.35,0.35);
\filldraw (-1,0) circle (2pt);
\filldraw (0,0) circle (2pt);
\end{tikzpicture}
\end{minipage}}

\newcommand{\strazero}{\begin{minipage}{1.8cm}
\begin{tikzpicture}
\draw (0,0) -- (-1,0);
\draw (0,0) -- (0.35,-0.35);
\draw (-1,0) -- (-1.35,-0.35);
\draw (-0.5,0.35) node{ };
\filldraw (-1,0) circle (2pt);
\filldraw (0,0) circle (2pt);
\end{tikzpicture}
\end{minipage}}

\newcommand{\triastra}{\begin{minipage}{2.0cm}
\begin{tikzpicture}
\draw[ thick, -] (0,0.5) arc (90:450:0.5);
\draw [ultra thick, red, dotted] (-1,0) -- (-0.5,0);
\draw (0.5,0) -- (0.85,0.35);
\draw (0.5,0) -- (0.85,-0.35);
\filldraw (-0.5,0) circle (2pt);
\filldraw (0.5,0) circle (2pt);
\end{tikzpicture}
\end{minipage}}

\newcommand{\logm}{\begin{minipage}{2.5cm}
\centering
\begin{tikzpicture}
\draw[thick] (-0.5,0) arc (180:135:0.5);
\draw[thick] (-0.5,0) arc (180:405:0.5);
\draw [thick](-1,0) -- (-0.5,0);
\draw [ultra thick, red, dotted] (0,-0.5) -- (0,-1);
\draw [ultra thick, red, dotted] (-0.35,0.35) -- (-0.7,0.7);
\draw [thick] (1,0) -- (0.5,0);
\draw [ultra thick] (-0.35,0.35) .. controls (0, 0.6) .. (0.35,0.35);
\draw [ultra thick] (-0.35,0.35) .. controls (0, 0.1) .. (0.35,0.35);
\draw [ultra thick, white, dotted] (0,0.61) -- (0,1);
\draw [thick] (0,-0.5) .. controls (-0.17,-0.17) .. (-0.5,0);
\filldraw (-0.35,0.35) circle (2pt);
\filldraw (0,-0.5) circle (2pt);
\filldraw (-0.5,0) circle (2pt);
\filldraw (0.5,0) circle (2pt);
\filldraw (0.35,0.35) circle (2pt);
\end{tikzpicture}
\end{minipage}}

\newcommand{\logmcoup}{\begin{minipage}{1.5cm}
\centering
\begin{tikzpicture}
\draw[thick] (0,0.5) arc (90:450:0.5);
\draw [ thick] (0, 0.5) -- (0, 0.9);
\draw [ultra thick, red, dotted] (0, -0.5) -- (0, -0.9);
\draw [ultra thick, red, dotted] (0, 0.5) -- (-0.3, 0.8);
\draw [ thick] (0, 0.5) -- (0.3, 0.8);
\draw [thick] (0,0.5) -- (0,-0.5);
\filldraw (0,0.5) circle (2pt);
\filldraw (0,-0.5) circle (2pt);
\end{tikzpicture}
\end{minipage}}

\newcommand{\logmup}{\begin{minipage}{2.5cm}
\centering
\begin{tikzpicture}
\draw[thick] (0.5,0) arc (0:45:0.5);
\draw[thick] (-0.5,0) arc (180:135:0.5);
\draw [thick](-1,0) -- (-0.5,0);
\draw [ultra thick, red, dotted] (-0.35,0.35) -- (-0.7,0.7);
\draw [thick] (1,0) -- (0.5,0);
\draw [ultra thick] (-0.35,0.35) .. controls (0, 0.6) .. (0.35,0.35);
\draw [ultra thick] (-0.35,0.35) .. controls (0, 0.1) .. (0.35,0.35);
\draw (0,1) node{ };
\draw [ultra thick, red, dotted] (0.5,0) -- (0.5,-0.3);
\draw [ultra thick, red, dotted] (-0.5,0) -- (-0.5,-0.3);
\draw [ultra thick, red, dotted] (-0.5,0) -- (-0.3,-0.2);
\filldraw (-0.35,0.35) circle (2pt);
\filldraw (-0.5,0) circle (2pt);
\filldraw (0.5,0) circle (2pt);
\filldraw (0.35,0.35) circle (2pt);
\filldraw (0,-1) circle (0pt);
\end{tikzpicture}
\end{minipage}}

\newcommand{\logmcodown}{\begin{minipage}{1.5cm}
\centering
\begin{tikzpicture}
\draw[thick] (0,0.5) arc (90:450:0.5);
\draw [ thick] (0, -0.5) -- (0, -0.9);
\draw [ultra thick, red, dotted] (0, 0.5) -- (0, 0.9);
\draw [ultra thick, red, dotted] (0, -0.5) -- (-0.3, -0.8);
\draw [ thick] (0, -0.5) -- (0.3, -0.8);
\filldraw (0,0.5) circle (2pt);
\filldraw (0,-0.5) circle (2pt);
\end{tikzpicture}
\end{minipage}}

\newcommand{\logmdown}{\begin{minipage}{2.5cm}
\centering
\begin{tikzpicture}
\draw[thick] (0.5,0) arc (360:180:0.5);
\draw [thick](-1,0) -- (-0.5,0);
\draw [ultra thick, red, dotted] (0,-0.5) -- (0,-1);
\draw [ultra thick, red, dotted] (-0.35,0.75) -- (-0.35,0.45);
\draw [thick] (1,0) -- (0.5,0);
\draw [ultra thick] (-0.35,0.75) .. controls (0,1) .. (0.35,0.75);
\draw [ultra thick] (-0.35,0.75) .. controls (0,0.5) .. (0.35,0.75);
\draw [thick] (0,-0.5) .. controls (-0.17,-0.17) .. (-0.5,0);
\filldraw (0,-0.5) circle (2pt);
\filldraw (-0.5,0) circle (2pt);
\filldraw (0.5,0) circle (2pt);
\draw [ultra thick, red, dotted] (0.5,0) -- (0.5,0.3);
\draw [ultra thick, red, dotted] (-0.35,0.75) -- (-0.5,1.1);
\draw [ultra thick, red, dotted] (0.35,0.75) -- (0.35,0.45);
\draw [ultra thick, red, dotted] (-0.5,0) -- (-0.5,0.3);
\draw [ultra thick, red, dotted] (0.5,0) -- (0.3,0.2);
\filldraw (-0.35,0.75) circle (2pt);
\filldraw (0.35,0.75) circle (2pt);
\filldraw (0,-1) circle (0pt);
\end{tikzpicture}
\end{minipage}}

\newcommand{\logthn}{\begin{minipage}{3cm}
\begin{tikzpicture}
\draw[thick, -] (0,0.5) arc (90:450:0.5);
\draw [thick](-1,0) -- (-0.5,0);
\draw [ultra thick, red, dotted] (0,-0.5) -- (0,-1);
\draw [ultra thick, red, dotted] (0,0.5) -- (0,1);
\draw [thick] (1,0) -- (0.5,0);
\draw [thick] (0,0.5) .. controls (0.17,0.17) .. (0.5,0);
\draw [thick] (0,-0.5) .. controls (-0.17,-0.17) .. (-0.5,0);
\filldraw (0,0.5) circle (2pt);
\filldraw (0,-0.5) circle (2pt);
\filldraw (-0.5,0) circle (2pt);
\filldraw (0.5,0) circle (2pt);
\draw (0,1.2) node{${\scriptstyle 4}$};
\draw (1.2,0) node{${\scriptstyle 3}$};
\draw (0,-1.2) node{${\scriptstyle 2}$};
\draw (-1.2,0) node{${\scriptstyle 1}$};
\end{tikzpicture}
\end{minipage}}

\newcommand{\logth}{\begin{minipage}{2.5cm}
\begin{center}
\begin{tikzpicture}
\draw[thick, -] (0,0.5) arc (90:450:0.5);
\draw [thick](-1,0) -- (-0.5,0);
\draw [ultra thick, red, dotted] (0,-0.5) -- (0,-1);
\draw [ultra thick, red, dotted] (0,0.5) -- (0,1);
\draw [thick] (1,0) -- (0.5,0);
\draw [thick] (0,0.5) .. controls (0.17,0.17) .. (0.5,0);
\draw [thick] (0,-0.5) .. controls (-0.17,-0.17) .. (-0.5,0);
\filldraw (0,0.5) circle (2pt);
\filldraw (0,-0.5) circle (2pt);
\filldraw (-0.5,0) circle (2pt);
\filldraw (0.5,0) circle (2pt);
\end{tikzpicture}
\end{center}
\end{minipage}}
\newcommand{\logthup}{\begin{minipage}{2.5cm}
\centering
\begin{tikzpicture}
\draw[thick, -] (0.5,0) arc (0:180:0.5);
\draw [thick](-1,0) -- (-0.5,0);
\draw [white, thick] (0,-0.5) -- (0,-1);
\draw [ultra thick, red, dotted] (0,0.5) -- (0,1);
\draw [thick] (1,0) -- (0.5,0);
\draw [thick] (0,0.5) .. controls (0.17,0.17) .. (0.5,0);
\draw [ultra thick, red, dotted] (0.5,0) -- (0.5,-0.3);
\draw [ultra thick, red, dotted] (-0.5,0) -- (-0.5,-0.3);
\draw [ultra thick, red, dotted] (-0.5,0) -- (-0.3,-0.2);
\filldraw (0,0.5) circle (2pt);
\filldraw (-0.5,0) circle (2pt);
\filldraw (0.5,0) circle (2pt);
\end{tikzpicture}
\end{minipage}}
\newcommand{\coup}{ \begin{minipage}{1.3cm}
\begin{tikzpicture}
\draw[color=black, thick, -] (0,0.5) arc (90:450:0.5);
\draw [ thick] (0, 0.5) -- (0, 0.9);
\draw [ultra thick, red, dotted] (0, -0.5) -- (0, -0.9);
\draw [ultra thick, red, dotted] (0, 0.5) -- (-0.3, 0.8);
\draw [ thick] (0, 0.5) -- (0.3, 0.8);
\draw [black,  thick] (0,0.5) -- (0,-0.5);
\filldraw (0,0.5) circle (2pt);
\filldraw (0,-0.5) circle (2pt);
\end{tikzpicture}
\end{minipage}}
\newcommand{\logthdown}{\begin{minipage}{2.5cm}
\centering
\begin{tikzpicture}
\draw[thick, -] (0.5,0) arc (360:180:0.5);
\draw [thick](-1,0) -- (-0.5,0);
\draw [white, thick] (0,0.5) -- (0,1);
\draw [ultra thick, red, dotted] (0,-0.5) -- (0,-1);
\draw [thick] (1,0) -- (0.5,0);
\draw [ultra thick, red, dotted] (0.5,0) -- (0.5,0.3);
\draw [ultra thick, red, dotted] (-0.5,0) -- (-0.5,0.3);
\draw [ultra thick, red, dotted] (0.5,0) -- (0.3,0.2);
\draw [thick] (0,-0.5) .. controls (-0.17,-0.17) .. (-0.5,0);
\filldraw (0,-0.5) circle (2pt);
\filldraw (-0.5,0) circle (2pt);
\filldraw (0.5,0) circle (2pt);
\end{tikzpicture}
\end{minipage}}
\newcommand{\codown}{ \begin{minipage}{1.3cm}
\begin{tikzpicture}
\draw[color=black, thick, -] (0,0.5) arc (90:450:0.5);
\draw [ thick] (0, -0.5) -- (0, -0.9);
\draw [ultra thick, red, dotted] (0, 0.5) -- (0, 0.9);
\draw [ultra thick, red, dotted] (0, -0.5) -- (-0.3, -0.8);
\draw [ thick] (0, -0.5) -- (0.3, -0.8);
\draw [black,  thick] (0,0.5) -- (0,-0.5);
\filldraw (0,0.5) circle (2pt);
\filldraw (0,-0.5) circle (2pt);
\end{tikzpicture}
\end{minipage}}